\documentclass[letterpaper]{article} 
\usepackage{aaai25}  
\usepackage{times}  
\usepackage{helvet}  
\usepackage{courier}  
\usepackage[hyphens]{url}  
\usepackage{graphicx} 
\usepackage{natbib}  
\usepackage{caption} 
\usepackage{algorithm}
\usepackage{algorithmic}
\usepackage{csquotes}
\usepackage{booktabs}

\usepackage{amsthm}
\usepackage{multirow}

\usepackage{makecell}

\usepackage{newfloat}
\usepackage{listings}
\DeclareCaptionStyle{ruled}{labelfont=normalfont,labelsep=colon,strut=off} 
\floatstyle{ruled}
\newfloat{listing}{tb}{lst}{}
\floatname{listing}{Listing}
\newtheorem{theorem}{Theorem}
\newtheorem{definition}{Definition}
\newtheorem{lemma}{Lemma}
\newtheorem{corollary}{Corollary}
\usepackage{amssymb}

\title{Combining Priors with Experience: Confidence Calibration Based on Binomial Process Modeling}
\author {
	Jinzong Dong,
	Zhaohui Jiang\thanks{corresponding author.},
	Dong Pan,
	Haoyang Yu
}
\affiliations {
	School of Automation, Central South University\\
	\{dongjinzong, jzh0903, pandong, yu.haoyang\}@csu.edu.cn
}

\usepackage{bibentry}

\begin{document}

\maketitle

\begin{abstract}
Confidence calibration of classification models is a technique to estimate the true posterior probability of the predicted class, which is critical for ensuring reliable decision-making in practical applications. Existing confidence calibration methods mostly use statistical techniques to estimate the calibration curve from data or fit a user-defined calibration function, but often overlook fully mining and utilizing the prior distribution behind the calibration curve. However, a well-informed prior distribution can provide valuable insights beyond the empirical data under the limited data or low-density regions of confidence scores. To fill this gap, this paper proposes a new method that integrates the prior distribution behind the calibration curve with empirical data to estimate a continuous calibration curve, which is realized by modeling the sampling process of calibration data as a binomial process and maximizing the likelihood function of the binomial process. We prove that the calibration curve estimating method is Lipschitz continuous with respect to data distribution and requires a sample size of $3/B$ of that required for histogram binning, where $B$ represents the number of bins. Also, a new calibration metric ($TCE_{bpm}$), which leverages the estimated calibration curve to estimate the true calibration error (TCE), is designed. $TCE_{bpm}$ is proven to be a consistent calibration measure. Furthermore, realistic calibration datasets can be generated by the binomial process modeling from a preset true calibration curve and confidence score distribution, which can serve as a benchmark to measure and compare the discrepancy between existing calibration metrics and the true calibration error. The effectiveness of our calibration method and metric are verified in real-world and simulated data. We believe our exploration of integrating prior distributions with empirical data will guide the development of better-calibrated models, contributing to trustworthy AI.
\end{abstract}

%
 \begin{links}
 \link{Code}{https://github.com/NeuroDong/TCEbpm}
 \link{Extended version}{https://arxiv.org/abs/2412.10658}
 \end{links}

\section{Introduction}
The prediction accuracy of modern machine learning classification methods such as deep neural networks is steadily increasing, leading to adoption in many safety-critical fields such as intelligent transportation \cite{lu2024disentangling}, industrial automation \cite{zidonghua}, and medical diagnosis \cite{luo2024knowledge}. However, decision-making systems in these fields not only require high accuracy but also require signaling when they might be wrong \cite{munir2024cal}. For example, in an automatic disease diagnosis system, when the confidence of the diagnostic model is relatively low, the decision-making should be passed to the doctor \cite{jiang2012calibrating}. Specifically, along with its prediction, a classification model should offer accurate confidence (matching the true probability of event occurrence). In addition, accurate confidence also provides more detailed information than the no-confidence or class label \cite{huang2020experimental}. For example, doctors can gather more information to make more reliable decisions in “there is a 70\% probability that the patient has cancer” than just a class label of “cancer”. Furthermore, accurate confidence facilitates the incorporation of classification models into other probabilistic models. For instance, accurate confidence allows active learning to select more representative samples \cite{han2024balque} and improves the generalization performance of knowledge distillation \cite{li2023distilling}. Therefore, pursuing more accurate confidence for classification models is a significant work \cite{penso2024confidence,wang2024moderate}.

However, modern classification neural networks often suffer from inaccurate confidence \cite{guo2017calibration}, which means that their confidence does not match the true probabilities of predicted class. For example, if a deep neural network classifies a medical image as \enquote{benign} with a confidence score of 0.99, the true probability of the medical image being \enquote{benign} could be significantly lower than 0.99, and even its true class may be \enquote{malignant}. Therefore, in recent years, this problem has been attracting increasing attention \cite{dong2024survey,geng2024survey}, and many confidence calibration methods, which aim to obtain more accurate confidence through additional processing, have been proposed \cite{silva2023classifier,zhang2023survey}.

Previous works calibrate confidence mostly from three directions: 1) Performing calibration during the classifier's training (train-time calibration), usually modifying the classifier's objective function \cite{classAdaptive,muller2019does,fernando2021dynamically}; 2) Binning confidence scores and estimating the calibrated confidence using the average accuracy inside the bins (binning-based calibration) \cite{zadrozny2001obtaining,naeini2015obtaining,patelmulti}; 3) Fitting a function on logit or confidence score so that the outcome of the function is calibrated (fitting-based calibration) \cite{platt1999probabilistic,guo2017calibration,zadrozny2002transforming,mixNmatch}. Despite the existence of the three valuable methods mentioned above, these methods mainly focus on how to estimate calibrated confidence from data or fit a user-defined calibration function (e.g., temperature scaling \cite{guo2017calibration} is a scaling function of logit, and platt scaling \cite{platt1999probabilistic} is a sigmoid function of logit), without systematically and principledly utilizing the prior distributions behind the calibration curve. However, in the field of statistics, especially bayesian statistics, the utilization of prior distributions is crucial. In particular, when data size is insufficient, such as in low-density regions of confidence scores, a correct prior distribution can be more informative than the data. 

Therefore, a natural but ignored question is studied: how to integrate the prior distribution behind the calibration curve with the empirical data to achieve a better calibration and develop a more accurate calibration metric? To address this, this paper conducts binomial process modeling on the sampling process of calibration data, which cleverly integrates the prior distribution of the calibration curve with the empirical data. By maximizing the likelihood function of the binomial process, a continuous calibration curve can be estimated. A general and effective prior is suggested as the prior distribution behind calibration curves, which is a principled function family derived from beta distributions. We prove that the estimated calibration curve is Lipschitz continuous with respect to data distribution and requires only a sample size of $3/B$ of that required for histogram binning, where $B$ represents the number of bins. Furthermore, using the estimated calibration curve, a new calibration metric is proposed, named $TCE_{bpm}$. $TCE_{bpm}$ is proved to be a consistent calibration measure \cite{blasiok2023unifying}. Finally, by modeling the sampling process of the calibration data as a binomial process, we can sample realistic calibration data from a preset true calibration curve and confidence score distribution, which can serve as a benchmark to measure and compare the discrepancy between existing calibration metrics and the true calibration error.

Our contributions can be summarized as follows:
\begin{itemize}
	\item A new calibration curve estimating method is proposed, which integrates prior distributions behind the calibration curve with empirical data through binomial process modeling. By maximizing the likelihood function of the binomial process, a continuous calibration curve can be estimated. We prove that the new calibration curve estimation method is Lipschitz continuous with respect to the data distribution and requires only a sample size of $3/B$ of that required for histogram binning, where $B$ represents the number of bins.
	\item A new calibration metric ($TCE_{bpm}$) is proposed, which leverages the estimated calibration curve to estimate the true calibration error (TCE). Theoretically, $TCE_{bpm}$ is proven to be a consistent calibration measure.
	\item A realistic calibration data simulation method based on binomial process modeling is proposed, which can serve as a benchmark to measure and compare the discrepancy between existing calibration metrics and the true calibration error.
\end{itemize}

\section{Background and Related Work}
\label{Background}
For a $K$-class classification problem, let $(X,Y)\in \mathcal{X} \times \mathcal{Y}$ be jointly distributed random variables, where $\mathcal{X} \subset R^{d}$ denotes the feature space and $\mathcal{Y}=\{1,2,...,K\}$ is the label space. The classification model can be expressed as $f(X): \mathcal{X} \to \mathcal{S}$, where $S=(S_{1},S_{2},...,S_{K}) \in \mathcal{S} \subset \Delta _{K-1}$ and $\Delta _{K-1}$ represents a simplex with free-degree $K-1$. The predicted class $\hat Y = \mathop {{\mathop{\rm argmax}\nolimits} }\limits_k {\{ {S_k}\} _{1 \le k \le K}}$, and the confidence score of predicted class is $\hat S = \max {\{ {S_k}\} _{1 \le k \le K}}$.

Typically, we just care about the confidence of the predicted class. In this case, a multi-classification problem can be formally unified into a binary classification problem. Let the \enquote{hit} variable $H=I[Y = \hat Y]$, where $I$ is the indicative function, that is, when $Y = \hat Y$, $I[Y = \hat Y]=1$, otherwise $I[Y = \hat Y]=0$. Therefore, the data samples become observations of $(\hat S,H)$.
\subsection{Confidence Calibration}
The purpose of confidence calibration is to make the confidence of predicted class match the true posterior probability of the predicted class. Formally, we state:
\begin{definition}
	\textnormal{\textbf{(Perfect calibration)}} A classification model is perfectly calibrated if the following equation is satisfied:
	\begin{equation}
		P(Y = \hat Y|\hat S = \hat s) = \hat s,
		\label{top_calibrated_eq}
	\end{equation}
	where $\hat s $ is the observed confidence score of predicted class, $\hat Y$ is the predicted class.
	\label{top_calibration}
\end{definition}

Obviously, Eq. \ref{top_calibrated_eq} can also be written as $P(H = 1|\hat S = \hat s) = \hat s$. Typically, we call $P(H=1|\hat S)$ as the true calibration curve.

\subsection{Estimates of Calibration Curve}
Currently, confidence calibration methods can be mainly divided into two groups: \textbf{train-time} calibration \cite{classAdaptive,muller2019does,fernando2021dynamically} and \textbf{post-hoc} calibration \cite{guo2017calibration, dirichlet, mixNmatch, intraOderPreserving}. Train-time calibration usually performs calibration during the training of the classifier by modifying the objective function, which may increase the computational cost of the classification task \cite{naeini2015obtaining} and affect the classification effect \cite{sampleDependent}. Post-hoc calibration learns a transformation (referred to as a calibration map) of the trained classifier’s predictions on a calibration dataset in a post-hoc manner \cite{mixNmatch}, which does not change the weights of the classifier and usually performs simple operations.

Pioneering work along the post-hoc calibration direction can be divided into two subgroups: \textbf{binning-based} calibration and \textbf{fitting-based} calibration. Binning-based calibration methods divide the confidence scores into multiple bins and estimate the calibrated value using the average accuracy inside the bins. The classic methods include Histogram binning \cite{zadrozny2001obtaining}, Bayesian binning \cite{naeini2015obtaining}, Mutual-information-maximization-based binning \cite{patelmulti}. Fitting-based calibration methods fit a function on logit or confidence score so that the outcome of the function is calibrated. The classic methods include Platt scaling \cite{platt1999probabilistic}, Temperature scaling \cite{guo2017calibration}, Isotonic regression \cite{zadrozny2002transforming}, Mix-n-Match \cite{mixNmatch}.

\subsection{Estimates of Calibration Error}
\subsubsection{True Calibration Error}
The true calibration error is described as the $l_{p}$ norm difference between the confidence score of the predicted class and the true likelihood of being correct \cite{kumar2019verified}:
\begin{equation}
	TCE = {({E_{\hat S}}[|\hat S - P(H = 1|\hat S){|^p}])^{\frac{1}{p}}}.
	\label{TCE}
\end{equation}
The true calibration curve $P(H = 1|\hat S)$ and the distribution of confidence scores $\hat S \sim \hat \mathcal{S}$ determine the value of TCE. TCE is not computable since the ground truth of $P(H = 1|\hat S)$ and the true distribution of $\hat S$ cannot be obtained in practice. Therefore, statistical methods are needed to estimate $P(H = 1|\hat S)$ and the distribution of $\hat \mathcal{S}$, and then estimate the true calibration error.
\subsubsection{Binning-Based Calibration Metrics}
Binning-based calibration metrics use the average accuracy of each bin to approximate $P(H=1|\hat S)$ and use the sample size proportion of the bin to approximate the $\hat \mathcal{S}$. Formally, assume that all confidence scores are partitioned into $M$ equally-spaced non-overlapping bins, and the $i$-th bin is represented by $B_{i}$, then the binning-based expected calibration error ($ECE_{bin}$) is calculated as follows:
\begin{equation}
	EC{E_{bin}} = {(\sum\limits_{i = 1}^M {\frac{{\left| {{B_i}} \right|}}{N}{{\left| {{\rm{acc}}({B_i}) - {\rm{conf}}({B_i})} \right|}^p}} )^{\frac{1}{p}}},
	\label{ECE}
\end{equation}
where $N$ represents the total number of samples, $|B_{i}|$ represents the element count of $B_{i}$, ${\rm acc}(B_{i})$ represents the average accuracy on $B_{i}$, and ${\rm conf}(B_{i})$ represents the average confidence on $B_{i}$. Typically, the binning scheme is divided into equal width binning \cite{guo2017calibration, naeini2015obtaining} and equal mass binning \cite{kumar2019verified,zadrozny2001obtaining}. Recently, $Nixon\ et\ al.$ \cite{nixon2019measuring} and $Roelofs\ et\ al.$ \cite{roelofs2022mitigating} observed that $ECE_{bin}$ with equal mass binning produces more stable calibration effect. $ECE_{bin}$ is sensitive to the binning scheme \cite{kumar2019verified,nixon2019measuring}. Therefore, some improvements to $ECE_{bin}$ have been proposed. $Ferro$ and $Fricker$ \cite{ferro2012bias} and $Brocker$ \cite{brocker2012estimating} propose a debiased estimator, $ECE_{debiased}$, which employs a jackknife technique to estimate the per-bin bias in the standard $ECE_{bin}$. This bias is then subtracted to estimate the calibration error better. $Roelofs\ et\ al.$ \cite{roelofs2022mitigating} propose $ECE_{sweep}$, which introduces the monotonically increasing property of the calibration curve into $ECE_{bin}$.
\subsubsection{Binning-Free Calibration Metrics}
In recent years, confidence calibration evaluation methods that are not based on binning have also been proposed. $Gupta\ et\ al.$ \cite{gupta2020calibration} proposed $KS-error$, which uses the Kolmogorov-Smirnov statistical test to evaluate the calibration error. $Zhang\ et\ al.$ \cite{mixNmatch} and $Błasiok\ et\ al.$ \cite{blasioksmooth} propose smoothed Kernel Density Estimation (KDE) methods for evaluating calibration error. $Chidambaram\ et\ al.$ \cite{chidambaramflawed} smooth the logit and then use the smoothed logit to build calibration metric.
\subsection{Combining Prior with Experience}
In statistics, integrating prior distribution with experience data to estimate the distribution behind the data is a classic and practical tradition \cite{zellner1996models,lavine1991sensitivity}. Priors refer to initial inference on the form or value of model parameters before observing data. Experience refers to the knowledge a model learns from data. When there is enough data, the model can learn well from experience, and the role of the prior may not be reflected. However, when data size is insufficient, a well-informed prior is often more effective than experience data.

In confidence calibration, most existing calibration methods predominantly focus on how to estimate calibrated confidence from data, fit a user-defined calibration function on logit or confidence score, or use naive fitting (e.g., least square method, minimizes cross-entropy loss) to combine priors (e.g., beta prior \cite{kull2017beyond}, dirichlet prior \cite{kull2019beyond}) with experience. However, although fitting a user-defined calibration function may also imply some user-observed priors, such as the choice of function shape, these priors are too empirical, and their universality needs to be considered. In addition, naive fitting is prone to be overly affected by data with larger statistical biases (e.g., sparse data). Therefore, it is necessary to study a principled method that better integrates a well-informed prior distribution with empirical data to estimate the calibration curve. 
\section{Method}
In this section, the following questions are studied: 1) How to introduce priors to estimate calibration curve $P(H = 1|\hat S)$ better? 2) How to choose an appropriate prior? 3) How to build a calibration metric using the estimated calibration curve? 4) How about the theoretical properties of the proposed method?

In Section \ref{BPM}, to solve the first problem, the sampling process of the calibration data is modeled as a binomial process, and then the calibration curve can be estimated by maximizing the likelihood function of this binomial process. In Section \ref{SP}, a general and effective prior function family is suggested. In Section \ref{TCE_bpm}, a new calibration metric $TCE_{bpm}$ is proposed. Section \ref{TG} analyzes the theoretical guarantee of the proposed calibration method and metric.

\subsection{Estimating Calibration Curve}
\label{BPM}
\subsubsection{Binomial Process} 
For any fixed $\hat S$, the repeated sampling process of $H$ is a binomial distribution on $N_{^{\hat S}}^{pos}$, where $N_{^{\hat S}}^{pos} = \sum\nolimits_{i = 1}^{{N_{\hat S}}} {H_i^{(\hat S)}}$ represents the number of \enquote{hit}, ${H_i^{(\hat S)}}$ represents the $i$-th \enquote{hit} label at $\hat S$, ${{N_{\hat S}}}$ represents the number of samples at $\hat S$. Specifically, $N_{^{\hat S}}^{pos} \sim \mathcal{BI}({N_{\hat S}},P(H = 1|\hat S))$, where $\mathcal{BI(\cdot)}$ represents binomial distribution. Formally, the following equation is satisfied:
\begin{equation}
	\begin{array}{l}
		P(N_{\hat S}^{pos}) = {\cal B}{\cal I}({N_{\hat S}},P(H = 1|\hat S))\\
		= C_{{N_{\hat S}}}^{N_{\hat S}^{pos}}P{(H = 1|\hat S)^{N_{\hat S}^{pos}}}{(1 - P(H = 1|\hat S))^{{N_{\hat S}} - N_{\hat S}^{pos}}},
	\end{array}
\end{equation}
where $C_{{N_{\hat S}}}^{N_{\hat S}^{pos}}$ is the binomial coefficient. 

Furthermore, for all $\hat S \in [0,1]$, the repeated sampling process of $H$ is a binomial process. Binomial process is a random process \cite{grimmett2020probability} with binomial distribution in a continuous domain. Specifically, in the binomial process, the distribution of random variable $N_{^{\hat S}}^{pos}$ under every point $\hat S$ in the continuous domain $[0,1]$ is a binomial distribution. Since we are interested in $P(H=1|\hat S)$, $P(H=1|\hat S)$ is modeled as a prior function family $g(\hat S;\theta)$. Formally, the following equation is satisfied:
\begin{equation}
	\begin{array}{*{20}{l}}
		{P(N_{\hat S}^{pos}) = {\cal B}{\cal P}({N_{\hat S}},g(\hat S;\theta ))}\\
		{ = C_{{N}_{\hat S}}^{N_{\hat S}^{pos}}g{{(\hat S;\theta )}^{N_{\hat S}^{pos}}}{{(1 - g(\hat S;\theta ))}^{{N_{\hat S}} - N_{\hat S}^{pos}}}},
	\end{array}
\end{equation}
where $\mathcal{BP}(\cdot)$ represents binomial process, and $\hat S \in [0,1]$.

\subsubsection{Maximum Likelihood Estimation}
Our purpose is to estimate the calibration curve $g(\hat S;\theta)$. Usually, solving $g(\hat S;\theta)$ requires a combination of prior and experience. Prior defines that $g(\hat S;\theta)$ follows a function family of fixed parameter form or structure. Experience represents seeking the optimal parameter $\theta$ from the data. Maximum likelihood estimation is a classic and effective solution to this problem. Formally, the following equation needs to be solved:
\begin{equation}
	\mathop {{\rm{Argmax}}}\limits_\theta  P[D|g(\hat S;\theta )],
	\label{MHE}
\end{equation}
where $D$ is the calibration data set, and $D=\{ {\hat s_i},{h_i}\} _{1 \le i \le N}$. Equivalently, $D=\{(\hat s_{j},N_{\hat s_{j}},N_{\hat s_{j}}^{pos})\}_{1 \le j \le N^{'}}$, where $N^{'}$ is the number of sampling locations of $\hat S$. Therefore:
\begin{equation}
	\begin{array}{l}
		\mathop {{\rm{Argmax}}}\limits_\theta  P[D|g(\hat S;\theta )]\\
		= \mathop {{\rm{Argmax}}}\limits_\theta  {E_{\hat S \in D}}P[N_{\hat S}^{pos}|g(\hat S;\theta ),{N_{\hat S}}]\\
		= \mathop {{\rm{Argmax}}}\limits_\theta  {E_{\hat S \in D}}[{\cal B}{\cal P}({N_{\hat S}},g(\hat S;\theta ))]\\
		= \mathop {{\rm{Argmax}}}\limits_\theta  \sum\limits_{\hat S \in D} {{\cal B}{\cal P}({N_{\hat S}},g(\hat S;\theta ))P(\hat S)}.
	\end{array}
	\label{MHE}
\end{equation}
Therefore, Eq. \ref{MHE} can be solved just by knowing $P(\hat S)$. Typically, estimating $P(\hat S)$ requires discretization methods. In order to make $P(\hat S)$ as accurate as possible, the idea of bayesian averaging is adopted here. Formally, Eq. \ref{MHE} becomes that:
\begin{equation}
	\begin{array}{*{20}{c}}
		{\mathop {{\rm{Argmax}}}\limits_\theta  {E_B}[{E_b}[\mathcal{BP}({N_{{{\hat S}_b}}},g({{\hat S}_b};\theta ))|b]|B]},\\
		\Downarrow \\
		{\mathop {{\rm{Argmax}}}\limits_\theta  \sum\limits_{B \in \cal{B}} {P(B)\sum\limits_{b \in B} {P(b)} } \mathcal{BP}({N_{{{\hat S}_b}}},g({{\hat S}_b};\theta ))},
	\end{array}
	\label{MHE_bin}
\end{equation}
where $B \in \cal B$ represents $B$-th binning scheme in $D$, $\cal B$ represents the binning scheme space, $b$ represent the $b$-th bin in $B$, $\hat S_{b}$ represent the average confidence score in the $b$-th bin, ${N_{{{\hat S}_b}}}$ represent the sample size in the $b$-th bin. In this paper, a uniform prior for modeling $P(B)$ is used, and $P(b)=|b|/|D|$, where $|\cdot|$ the element count function.
\subsubsection{Equivalence Optimization} 
Directly solving Eq. \ref{MHE_bin} can obtain a feasible solution. However, due to the non-convexity of the likelihood function with respect to $g(\hat S;\theta )$ in Eq. \ref{MHE_bin}, the feasible solution may not be the global optimal solution. This non-global optimal solution will affect the estimation of the calibration curve. To solve this problem, an equivalent new optimization problem is proposed whose objective function is convex on $g(\hat S;\theta )$, as shown below:
\begin{equation}
	\mathop {{\rm{Argmin}}}\limits_\theta  \sum\limits_{B \in \cal{B}} {P(B)\sum\limits_{b \in B} {P(b)} }  \cdot {e^{{{(g({{\hat S}_b};\theta ) - \frac{{N_{{{\hat S}_b}}^{pos}}}{{{N_{{{\hat S}_b}}}}})^2}}}}.
	\label{equivalence}
\end{equation}
The proof of equivalence between Eq. \ref{MHE_bin} and Eq. \ref{equivalence} is provided in Appendix \ref{equivalence Solution}. In general, Eq. \ref{equivalence} can be solved well using a common optimizer (e.g., Gradient Descent \cite{andrychowicz2016learning}, Quasi-Newton Methods \cite{mokhtari2020stochastic}, Nelder-Mead \cite{nelder1965simplex}).

\subsection{Selecting Prior}
\label{SP}
The appropriate prior is crucial to solving Eq. \ref{equivalence}. If the prior function family is not correctly selected, the true calibration curve will not be in the feasible region (i.e., solution space). If the prior function family is correctly selected, a good calibration curve will be found quickly and efficiently. Next, this paper suggests a general and effective prior.

It is well known that the prior distribution of confidence scores can be modeled by beta distribution \cite{kull2017beyond,roelofs2022mitigating,kull2017beta}. Specifically, $P(\hat S|H=0)\sim {\cal B}eta(\alpha_{0},\beta_{0})$ and $P(\hat S|H=1)\sim {\cal B}eta(\alpha_{1},\beta_{1})$. By Bayes' theorem:
\begin{equation}
\begin{array}{*{20}{l}}
	{P(H = 1|\hat S) = \frac{{P(\hat S|H = 1)P(H = 1)}}{{P(\hat S|H = 1)P(H = 1) + P(\hat S|H = 0)P(H = 0)}}}\\
	{ = \frac{1}{{1 + \frac{{P(\hat S|H = 0)}}{{P(\hat S|H = 1)}} \cdot \frac{{P(H = 0)}}{{P(H = 1)}}}}}\\
	{ = \frac{1}{{1 + \frac{{beta{{({\alpha _0},{\beta _0})}^{ - 1}}{{\hat S}^{{\alpha _0} - 1}}{{(1 - \hat S)}^{{\beta _0} - 1}}}}{{beta{{({\alpha _1},{\beta _1})}^{ - 1}}{{\hat S}^{{\alpha _1} - 1}}{{(1 - \hat S)}^{{\beta _1} - 1}}}} \cdot \frac{{P(H = 0)}}{{P(H = 1)}}}}}\\
	{ = \frac{1}{{1 + \frac{{beta({\alpha _1},{\beta _1})}}{{beta({\alpha _0},{\beta _0})}} \cdot {{\hat S}^{{\alpha _0} - {\alpha _1}}}{{(1 - \hat S)}^{{\beta _0} - {\beta _1}}} \cdot \frac{{P(H = 0)}}{{P(H = 1)}}}}},
\end{array}
\label{prior_infer}
\end{equation}
where $beta(\cdot)$ represent beta function, ${\frac{{beta({\alpha _1},{\beta _1})}}{{beta({\alpha _0},{\beta _0})}} \cdot \frac{{P(H = 0)}}{{P(H = 1)}}}$ is a positive constant independent of $\hat S$. Let ${\frac{{beta({\alpha _1},{\beta _1})}}{{beta({\alpha _0},{\beta _0})}} \cdot\frac{{P(H = 0)}}{{P(H = 1)}}}$ be equal to $e^{c}$. In order to maintain the monotonically increasing of the calibration curve \cite{roelofs2022mitigating,blasioksmooth}, the following two constraints need to be satisfied:
\begin{equation}
	\left\{ {\begin{array}{*{20}{c}}
			{{\alpha _0} - {\alpha _1} \le 0},\\
			{{\beta _0} - {\beta _1} \ge 0}.
	\end{array}} \right.
\end{equation}
Therefore, $g(\hat S;\theta)$ can be selected as:
\begin{equation}
	\begin{array}{l}
		g(\hat S;\alpha,\beta,c) = \frac{1}{{1 + {{\hat S}^{-\alpha}}{{(1 - \hat S)}^{\beta}} \cdot {e^c}}},
	\end{array}
	\label{prior_function}
\end{equation}
where $\hat S \in [0,1]$, $g(\hat S;\alpha,\beta,c) \in [0,1]$, $\alpha \ge 0$, $\beta \ge 0$, $c \in (-\infty,+\infty)$.

To sum up, the computational steps for estimating the calibration curve $P(H = 1|\hat S)$ are shown in Algorithm \ref{Estimating calibration curve alg}.
\begin{algorithm}[t]
	\caption{Estimating calibration curve.}
	\begin{algorithmic}
		\STATE \textbf{Initialize:} 
		\STATE \hspace{0.5cm}$P(D|g)=0$; $D=\{ {\hat s_i},{h_i}\} _{1 \le i \le N}^N$; $\cal B$; $\alpha$; $\beta$; $c$.
		\STATE \textbf{for} $B$ in $\cal B$:
		\STATE \hspace{0.5cm} \textbf{Initialize} $P(D|g,B)=0.$
		\STATE \hspace{0.5cm} \textbf{for} $b$ in $B$:
		\STATE \hspace{0.5cm}\hspace{0.5cm} $\hat S_{list} = \{\hat s_{i} \in b| D\}$; $H_{list} = \{h_{i} \in b| D\}$,
		\STATE \hspace{0.5cm}\hspace{0.5cm} $\hat S_{b}={\mathop{\rm mean}\nolimits} (\hat S_{list})$,
		\STATE \hspace{0.5cm}\hspace{0.5cm} $N_{\hat S_b}^{pos}={\mathop{\rm sum}\nolimits} (H_{list})$; ${N_{{{\hat S}_b}}} = {\mathop{\rm len}\nolimits}(H_{list})$,
		\STATE \hspace{0.5cm}\hspace{0.5cm} $P_{b} = {e^{{{(g({{\hat S}_b};\alpha,\beta,c) - \frac{{N_{{{\hat S}_b}}^{pos}}}{{{N_{{{\hat S}_b}}}}})^2}}}}$,
		\STATE \hspace{0.5cm}\hspace{0.5cm} $P(D|g,B) = P(D|g,B)+P_{b}\cdot \frac{{|b|}}{{|D|}}$,
		\STATE \hspace{0.5cm} $P(D|g) = P(D|g) + P(D|g,B)$,
		\STATE  $\alpha,\beta,c = \mathop {{\rm{argmin}}}\limits_{{\alpha},{\beta},c} P(D|g)$,
		\STATE \textbf{Return} $P(H=1|\hat S) = \frac{1}{{1 + {{\hat S}^{-\alpha}}{{(1 - \hat S)}^{\beta}} \cdot {e^c}}}$.
	\end{algorithmic}
	\label{Estimating calibration curve alg}
\end{algorithm}
\subsection{Estimating TCE}
\label{TCE_bpm}
As can be seen from Algorithm \ref{Estimating calibration curve alg}, the estimated calibration curve is a continuous function. This allows us to estimate the true calibration error (see Eq. \ref{TCE}) instead of just the expected calibration error (see Eq. \ref{ECE}).

According to Eq. \ref{TCE}, the calculation formula of TCE is:
\begin{equation}
	TCE = {[\int_0^1 {|P(H = 1|\hat S) - \hat S{|^p}} \xi (\hat S)d\hat S]^{\frac{1}{p}}},
\end{equation}
where $\xi (\hat S)$ is the probability density function of $\hat S$. Typically, $p$ can be set to 1. Since it is already known that the prior distribution of $\xi (\hat S)$ is the beta distribution, the parameters of the beta distribution can be estimated using the moment estimation method \cite{wang2006topics}, as shown in the Estimating TCE part of Algorithm \ref{Estimating TCE}. Therefore, TCE can be estimated by computing a definite integral on the interval [0,1], as shown in Algorithm \ref{Estimating TCE}.

\begin{algorithm}[t]
	\caption{Estimating TCE.}
	\begin{algorithmic}
		\STATE  $P(H=1|\hat S) = g(\hat S;\theta)$ from Algorithm \ref{Estimating calibration curve alg}.
		\STATE  $m={\mathop{\rm mean}\nolimits}(\{\hat s_{i}|\hat s_{i} \in D \})$,
		\STATE  $v={\mathop{\rm var}\nolimits}(\{\hat s_{i}|\hat s_{i} \in D \})$,
		\STATE  $a_{1} = \frac{{{m^2}(1 - m)}}{v} - m$,
		\STATE  $a_{2} = a_{1}  \cdot \frac{{(1 - m)}}{m}$,
		\STATE  $\xi(\hat S) = beta(a_{1},a_{2})^{-1} \cdot {\hat S}^{(a_{1}-1)}\cdot (1-\hat S)^{a_{2}-1}$,
		\STATE  \textbf{Return} $TCE_{bpm} = \int_0^1 {|P(H=1|\hat S) - \hat S| \cdot \xi(\hat S)d} \hat S$.
	\end{algorithmic}
	\label{Estimating TCE}
\end{algorithm}

\subsection{Theoretical Guarantee}
\label{TG}
\subsubsection{Continuity}
Continuity with respect to data distribution is an important property for a calibration method and metric. It tells us whether a slight change in data distribution will lead to a drastic jump in the calibration curve and the calibration metric. Before conducting continuity analysis, a distance measure of the data distributions needs to be defined. It tells us how far away the two distributions are. In this paper, Wasserstein distance is used, as shown in Definition \ref{Wasserstein}. Wasserstein distance measures the minimum cost of transforming one distribution into another.
\begin{definition}
	\textnormal{\textbf{(Wasserstein distance)}} For two data distribution $D_{1},D_{2}$ over $[0,1] \times \{ 0,1\}$, let $\Gamma$ be the family of all couplings of distributions $D_{1}$ and $D_{2}$, Wasserstein distance is defined as follows:
	\begin{equation}
		W({D_1},{D_2}) = \mathop {\inf }\limits_{\gamma  \in \Gamma } \mathop E\limits_{({{\hat S}_1},{{\hat S}_2}) \sim \gamma } [|\frac{{N_{{{\hat S}_1}}^{pos}}}{{{N_{{{\hat S}_1}}}}} - \frac{{N_{{{\hat S}_2}}^{pos}}}{{{N_{{{\hat S}_2}}}}}| + |{\hat S_1} - {\hat S_2}|].
	\end{equation}
	Specially, when $N_{\hat S_{1}}=N_{\hat S_{2}}=1$, then:
	\begin{equation}
		W({D_1},{D_2}) = \mathop {\inf }\limits_{\gamma  \in \Gamma } \mathop E\limits_{({{\hat S}_1},{{\hat S}_2}) \sim \gamma } [|{H_1} - {H_2}| + |{{\hat S}_1} - {{\hat S}_2}|].
	\end{equation}
	\label{Wasserstein}
\end{definition}
Next, this paper first proves that $g(\hat S;\theta)$ obtained by Algorithm \ref{Estimating calibration curve alg} is Lipschitz continuity $w.r.t.$ data distributions, as shown in Theorem \ref{calibration_curve_lipschitz}. Then, this paper proves that $TCE_{bpm}$ is Lipschitz continuity $w.r.t.$ data distribution when certain conditions are met, as shown in Theorem \ref{lipschitz}. The proofs of Theorem \ref{calibration_curve_lipschitz} and Theorem \ref{lipschitz} are given in Appendix \ref{lipschitz proof}.
\begin{theorem}
	For two distribution $D_{1},D_{2}$ over $[0,1] \times \{ 0,1\}$, let $\Gamma$ be the family of all couplings of distributions $D_{1}$ and $D_{2}$, and $g({{\hat S}};{\theta _{{D}}})$ represents the calibration curve learned from $D$ via Eq. \ref{equivalence}, then $\forall \gamma \in \Gamma$, it holds that:
	\begin{equation}
		\begin{array}{l}
			\mathop {{\rm{ }}E}\limits_{({{\hat S}_1},{{\hat S}_2}) \sim \gamma } |g({{\hat S}_1};{\theta _{{D_1}}}) - g({{\hat S}_2};{\theta _{{D_2}}})|\\
			\le L \cdot \mathop E\limits_{({{\hat S}_1},{{\hat S}_2}) \sim \gamma } [|\frac{{N_{{{\hat S}_1}}^{pos}}}{{{N_{{{\hat S}_1}}}}} - \frac{{N_{{{\hat S}_2}}^{pos}}}{{{N_{{{\hat S}_2}}}}}| + |{{\hat S}_1} - {{\hat S}_2}|],
		\end{array}
	\end{equation}
	where $L \ge 0$. Therefore:
	\begin{equation}
		\mathop {\inf }\limits_{\gamma  \in \Gamma } \mathop {{\rm{ }}E}\limits_{({{\hat S}_1},{{\hat S}_2}) \sim \gamma } |g({{\hat S}_1};{\theta _{{D_1}}}) - g({{\hat S}_2};{\theta _{{D_2}}})| \le L \cdot W({D_1},{D_2}).
	\end{equation}
	\label{calibration_curve_lipschitz}
\end{theorem}
\begin{theorem}
	$\forall \hat S \in [0,1]$, if $g(\hat S; \theta_{D})$ and $\xi_{D}(\hat S)$ are Lipschitz continuous $w.r.t.$ $D$, then for two distribution $D_{1},D_{2}$ over $[0,1] \times \{ 0,1\}$, $TCE_{bpm}$ satisfies:
	\begin{equation}
		|TC{E_{bpm}}({D_1}) - TC{E_{bpm}}({D_2})| \le L \cdot W({D_1},{D_2}),
	\end{equation}
	where $L \ge 0$.
	\label{lipschitz}
\end{theorem}
\subsubsection{Consistency}
To theoretically prove the effectiveness of a calibration metric, $Błasiok\ et\ al.$ \cite{blasiok2023unifying} have proposed a unified theoretical framework: consistent calibration measure. Consistent calibration measure means two things: 1) When the true distance to calibration is small, the calibration metric should also be small(i.e., Robust completeness); 2) When the true distance to calibration is large, the calibration metric should also be large (i.e., Robust soundness). 

Before defining the consistent calibration measure, we need to define how far a data distribution $D$ is from its nearest perfect calibration distribution, as shown in the Definition \ref{True distance to calibration}. Then, the consistent calibration measure can be defined as shown in Definition \ref{Consistency}.

Theorem \ref{Consistenty} proves that $TCE_{bpm}$ is a consistent calibration measure when certain conditions are met. Corollary \ref{corollary} proves that when $P(\hat S|H=0)$ and $P(\hat S|H=1)$ follow beta distribution, $TCE_{bpm}$ calculated using Eq. \ref{prior_function} is a consistent calibration measure. The proof of Theorem \ref{Consistenty} is given in Appendix \ref{Consistenty proof}, and the proof of Corollary \ref{corollary} is given in Appendix \ref{Corollary proof}.
\begin{definition}
	\textnormal{\textbf{(True distance to calibration)}}
	$\forall D$ over $[0,1] \times \{ 0,1\}$, let $\cal{P}$ be the family of all perfectly calibrated distributions, the true distance to calibration is:
	\begin{equation}
		\underline{dCE}(D) = \mathop {\inf }\limits_{\mathcal{D} \in \mathcal{P}} W(D,\mathcal{D}).
	\end{equation}
	
	\label{True distance to calibration}
\end{definition}

\begin{definition}
	\textnormal{\textbf{(Consistent calibration measure)}}
	For $q,t,L_{1},L_{2} > 0$, calibration metric $\mu$, and data distribution $D$ over $[0,1]\times\{0,1\}$, if:
	\begin{equation}
		\mu (D) \le {L_1} \cdot {({\underline{dCE}(D)})^q},
	\end{equation}
	then $\mu$ satisfies $q$-robust completeness. If:
	\begin{equation}
		\mu (D) \ge {L_2} \cdot {({\underline{dCE}(D)})^t},
	\end{equation}
	then $\mu$ satisfies $t$-robust soundness. If satisfying robust completeness and soundness, $\mu$ is a consistent calibration measure.
	\label{Consistency}
\end{definition}

\begin{theorem}
$TCE_{bpm}$ is a consistent calibration measure if the following two conditions hold:
\begin{itemize}
\item The hypothesis set $\cal{G}$ (The set of all possible $g(\hat S; \theta_{D})$) includes the true calibration curve;
	
\item $\forall \hat S \in [0,1]$, if $g(\hat S; \theta_{D}) \in \cal{G}$ and $\xi_{D}(\hat S)$ are Lipschitz continuous $w.r.t.$ $D$.
\end{itemize}
\label{Consistenty}
\end{theorem}
\begin{corollary}
$\forall \hat S \in [0,1]$, if:
$$
g(\hat S;\alpha,\beta,c) = \frac{1}{{1 + {{\hat S}^{-\alpha}}{{(1 - \hat S)}^{\beta}} \cdot {e^c}}},
$$
and $\xi_{D}(\hat S)$ are Lipschitz continuous $w.r.t.$ $D$, and $P(\hat S|H=0)$ and $P(\hat S|H=1)$ follow beta distribution, $TCE_{bpm}$ is a consistent calibration measure.
\label{corollary}
\end{corollary}

\subsubsection{Sample Efficiency}
Sample efficiency tells us how many samples a method requires to keep the error small enough. According to Hoeffding’s inequality, the sample size required for histogram binning is $N\ge \frac{B\cdot \rm{In}(1/\delta)}{2\varepsilon^{2}}$, where $B$ represents the number of bins. Specifically, for the error of the histogram binning to be lower than $\varepsilon$, $\frac{{{\rm{In}}({\rm{1/}}\delta )}}{{2{\varepsilon ^2}}}$ samples are required at each bin. Next, we will prove that the sample efficiency of Algorithm \ref{Estimating calibration curve alg} is $N\ge \frac{3\cdot \rm{In}(1/\delta)}{2\varepsilon^{2}}$, as shown in Theorem \ref{efficiency}. This is due to the fact that the prior function of the calibration curve has three parameters, which ensures that the calibration curve can be uniquely determined by three distinct observation points. Therefore, in theory, as long as selecting three most representative bins, a good calibration curve can be estimated by Algorithm \ref{Estimating calibration curve alg}. The proof of Theorem \ref{efficiency} is given in Appendix \ref{efficiency_proof}.
\begin{theorem}
For the function family in Eq. \ref{prior_function}, if $P(\hat S|H=0)$ and $P(\hat S|H=1)$ follow beta distribution, then $\forall \delta \in (0,1)$, when $N\ge \frac{3\cdot \rm{In}(1/\delta)}{2\varepsilon^{2}}$, satisfy the following result with $1-\delta$ probability:
$$
E_{\hat S}|g(\hat S;\alpha_{D},\beta_{D},c_{D})-P(H=1|\hat S)| \le \varepsilon.
$$
\label{efficiency}
\end{theorem}
\begin{figure*}[h]
	\centering
	\includegraphics[width=1.\textwidth]{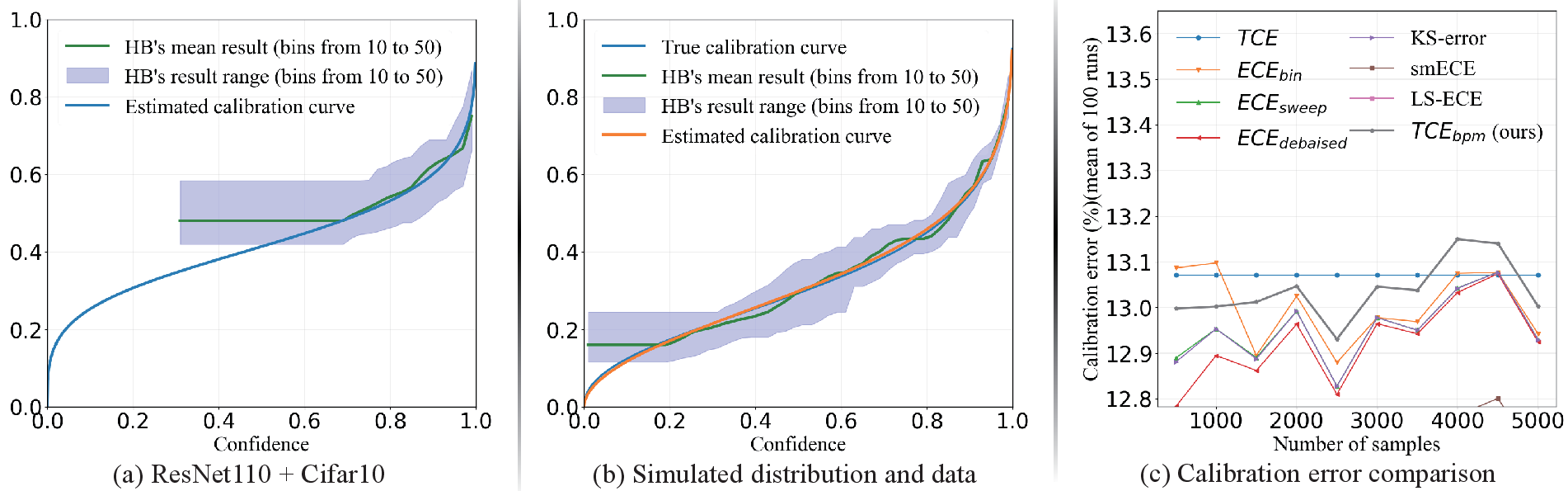}
	\caption{Experimental results of our method. HB represents Histogram binning \cite{zadrozny2001obtaining}. In (a), the estimated calibration curve on real data aligns well with histogram binning results from various binning schemes and closely matches the mean result. In (b), the calibration curve estimated by our method closely approximates the true calibration curve in simulated data. In (c), our calibration metric is closest to the true calibration error (TCE) in many times (e.g., when the number of samples is 1500, 2000, 2500, 3000, 3500, and 5000).}
	\label{Reliability_diagram}
\end{figure*}
\section{Simulating Datasets to Compare Evaluation Metrics}
\label{Simulating}
\begin{algorithm}[t]
	\caption{Simulating dataset with binomial process.}
	\begin{algorithmic}
		\STATE \textbf{Initialize:} 
		\STATE \hspace{0.5cm} $P(H=1|\hat S)=g(\hat S;\theta)$; $a_{1}$; $a_{2}$; $N$. 
		\STATE \textbf{Sampling:} 
		\STATE \hspace{0.5cm}$i=1$; $D=\{\}$.
		\STATE \hspace{0.5cm}\textbf{while $i \le N$:} 
		\STATE \hspace{0.5cm}\hspace{0.5cm} $\hat S= {\rm{Sampling}}\ {\rm{from}}\ Beta(a_{1},a_{2})$,
		\STATE \hspace{0.5cm}\hspace{0.5cm} $P(H=1|\hat S) = g(\hat S;\theta)$,
		\STATE \hspace{0.5cm}\hspace{0.5cm} $H= {\rm Sampling}\ {\rm from}\ {\cal BI}(1,P(H=1|\hat S))$,
		\STATE \hspace{0.5cm}\hspace{0.5cm} Adding $(\hat S,H)$ into $D$.
		\STATE \textbf{Return} $D$.
	\end{algorithmic}
	\label{alg2}
\end{algorithm} 
A key challenge in developing calibration metrics is the lack of ground truth for calibration curves and confidence scores, hindering the measurement of discrepancies between metrics and actual calibration errors. $Roelofs\ et\ al.$ \cite{roelofs2022mitigating} use the fitted function on the publicly available logit datasets as the true distribution behind the data, then use the fitted function to calculate TCE and compare TCE with existing calibration metrics. In this paper, an opposite operation is proposed, i.e., first preseting the true calibration distribution and then obtaining realistic calibration data through binomial process sampling.

Specifically, in Section \ref{BPM}, we model the process of sampling calibration data as a binomial process. Another important role of this modeling is that realistic calibration data sets can be sampled through using known calibration curves and confidence distributions, as shown in Algorithm \ref{alg2}. The confidence score $\hat S$ are first sampled, and then the calibration value $P(H=1|\hat S)$ is calculated. Then, $P(H=1|\hat S)$ is used as the probability of a single event success in the binomial distribution, and binomial distribution sampling is performed to sample $H$. Since the true calibration curve and confidence distribution are known, TCE can be calculated accurately. The sampled calibration data is then used to calculate other calibration metrics, and by comparing these metrics with the accurately calculated TCE, it can be determined which calibration metrics are better. 
\section{Results}
The effectiveness of the proposed method is verified from four perspectives: 1) On the real-world datasets, the calibration curve estimated by our method is compared with the results of histogram binning under various binning schemes; 2) On the datasets simulated by Algorithm \ref{alg2}, the discrepancy between the calibration curve estimated by our method and the true calibration curve is compared; 3) On the datasets estimated by Algorithm \ref{alg2}, the discrepancy between $TCE_{bpm}$ and the true calibration error is compared; 4) On the real-world datasets, multiple calibration metrics comparison between our calibraiton method with other calibration methods is performed. Due to space limitations, details of data selection and implementation details are given in Appendix \ref{Experimental data} and Appendix \ref{Hyperparameter settings}. In Appendix \ref{Experimental data}, ten publicly available logit datasets (i.e., real-world datasets) and five true distributions (named D1, D2,$\cdots$, and D5, respectively) were selected for the experiments.
\begin{table*}[h]
	\centering
	\footnotesize
	\setlength\tabcolsep{3pt}
	\renewcommand{\arraystretch}{1.3}
	\begin{tabular}{c|c|cccccc}
		\toprule
		Network and Dataset&Calibration methods&$ECE_{bin} \downarrow$&$ECE_{debaised} \downarrow$&$ECE_{sweep} \downarrow$&KS-error$\downarrow$&smECE$\downarrow$&$TCE_{bpm} \downarrow$\\
		\toprule
		\multirow{6}{*}{\makecell[c]{ResNet110\\ Cifar10}}&Uncalibration&0.04755&0.04752&0.04750&0.04750&0.04271& 0.05511\\
		&Temperature scaling&\underline{0.00745}&0.00568&0.00576&0.00590&\underline{0.00901}&0.00825\\
		&Isotonic regression&0.00753&\underline{0.00550}&0.00586&0.00601&0.00920&\underline{0.00649}\\
		&Mix-n-Match&\underline{0.00745}&0.00574&0.00576&0.00590&\underline{0.00901}&0.00825\\
		&Spline calibration&0.01322&0.01181&\underline{0.00347}&\underline{0.00430}&0.01280&0.01136\\
		&TPM calibration (Ours)&\textbf{0.00458}&\textbf{0.00312}&\textbf{0.00146}&\textbf{0.00162}&\textbf{0.00756}&\textbf{0.00214}\\
		\hline
		\bottomrule 
	\end{tabular}
	\caption{Comparison with other calibration methods on real data. Bold represents the best result, and underline represents the second-best result.}
	\label{Calibration methods comparation_one}
\end{table*}
\subsection{Estimated Results of Calibration Curves}
\subsubsection{Results in Real Datasets}
The estimated results of the calibration curve on the public ResNet110's logit dataset trained on Cifar10 is shown in (a) of Fig. \ref{Reliability_diagram}. The results on other datasets are shown in Appendix \ref{Comparison in real datasets in all confidence scores}. In order to intuitively show the effect of the calibration curve estimated by our method, the means and ranges of the calibration values estimated by the histogram binning under various binning schemes are simultaneously visualized. The calibration curve estimated by our method is close to the mean result of the histogram binning under various binnings and always falls within the result range of histogram binning under various binnings, indicating a relatively accurate and robust performance. In addition, Appendix \ref{Comparison in real datasets in all confidence scores} shows that our method can achieve such performance under various sharpness, meaning that our method has certain versatility. Furthermore, in the regions of low confidence scores, which are also the regions of low density of confidence scores, the calibration confidence estimated by the binning method fluctuates greatly (e.g., the result range is broad when the confidence score is lower than 0.8 in (a) of Fig. \ref{Reliability_diagram}) and sometimes even non-monotonic (e.g., (e), (g), and (h) in Fig. \ref{fig_fit_real_data} of Appendix \ref{Comparison in real datasets in all confidence scores}). The two situations are obviously unreasonable \cite{kumar2019verified,roelofs2022mitigating}. Thanks to the continuous and monotonic prior distribution, the two unreasonable situations can be well avoided by adopting our method.
\subsubsection{Results in Simulating Datasets}
The estimated result of the calibration curve on the dataset simulated by the true distribution D1 is shown in (b) of Fig. \ref{Reliability_diagram}. The results on other true distributions are shown in Appendix \ref{compare_metrics}. The calibration curve estimated by our method is closely aligned with the true calibration curve, with a mean absolute error (see Eq. \ref{EAD} in Appendix \ref{Hyperparameter settings}) of 0.0099, which is lower than 0.0233 of the mean result of histogram binning. This verifies the effectiveness of our method. In addition, in the regions of low confidence scores (i.e., the low-density regions of confidence scores), the broad result range indicates that the results estimated by the histogram binning method under a specific binning scheme may deviate from the true calibration value. Even the mean result of histogram binning under various binning schemes sometimes deviates significantly from the true calibration curve (e.g., when the confidence score is around 0.8 in D2 of Fig. \ref{fit_true_curve} in Appendix \ref{compare_metrics}). However, in our method, benefiting from the excellent integration between the well-informed prior distribution and the empirical data, the calibration curve can be well estimated even in the low-density regions of confidence scores. This good estimation effect makes us believe that fully integrating well-informed prior distribution with empirical data is a promising future direction for confidence calibration.
\subsection{Estimated Results of Calibration Metrics}
Six state-of-the-art calibration metrics are compared with $TCE_{bpm}$. The details of these six metrics are shown in Appendix \ref{Hyperparameter settings}. Among these calibration metrics, the one closer to TCE is more accurate. In (c) of Fig. \ref{Reliability_diagram}, the calculation results of these metrics on the dataset simulated by the true distribution D1 are shown. The results on other true distributions are shown in Appendix \ref{compare_metrics}. It can be seen that our calibration metric $TCE_{bpm}$ is closest to the true calibration error (TCE) in many times (e.g., when the number of samples is 1500, 2000, 2500, 3000, 3500, and 5000 in (c) of Fig. \ref{Reliability_diagram}). This shows that the calibration metrics estimated by our method are competitive. Even though we use a common confidence score estimation method (Moment estimation method \cite{wang2006topics}), the estimated $TCE_{bpm}$ is so competitive, which once again verifies that the calibration curve estimated by our method is quite accurate.
\subsection{Comparison with Other Calibration Methods}
Table \ref{Calibration methods comparation_one} shows the comparison results of calibration metrics between our calibration method and other calibration methods on the public ResNet110's logit dataset trained on Cifar10, where the last column is our calibration metric. All considered calibration methods can significantly improve confidence. The calibration error of our method on six metrics is 50.15\% less on average than the second place.
The comparison results on Wide-ResNet32's logits dataset of Cifar100 and DenseNet162 logits dataset of ImageNet are shown in Appendix \ref{effect of other calibration methods}.
\section{Conclusion and Discussion}
In this paper, we focus on how to effectively incorporate prior distributions behind calibration curves with empirical data to calibrate confidence better. To address this, we propose a new calibration curve estimation method via binomial process modeling and maximum likelihood estimation and perform theoretical analysis. Furthermore, using the estimated calibration curve, a new calibration metric is proposed, which is proven to be a consistent calibration measure. In addition, this paper proposes a new simulation method for calibration data through binomial process modeling, which can serve as a benchmark to measure and compare the discrepancy between existing calibration metrics and the true calibration error. Extensive empirical studies on real-world and simulated data support our findings and showcase the effectiveness of our method.
\subsubsection{Potential Impact, Limitations and Future Work}
We explore the impact of the prior distribution behind the calibration curve on the confidence calibration effect. We also provide a solution as a starting point to utilize the prior distribution of the calibration curve, which we believe has the potential to inspire more rich and subsequent works, ultimately leading to improved decision-making in real-world applications, especially for underrepresented populations and safety-critical scenarios. However, our study also has several limitations. First, the Bayesian average strategy is used to construct the likelihood function of the binomial process, which increases computational cost. Although the increased computational cost may be negligible, further reduction in computational cost is promising. Because Theorem \ref{efficiency} tells us that just three most representative bins are selected, a good parameter estimation can be achieved. Future research can investigate how to select the three most representative bins to achieve efficient parameter estimation. Additionally, we only focus on the confidence calibration in the closed-set classification problem. Future work can generalize this to multi-label, open-set, or generative settings.

\section*{Acknowledgements}
This work was supported in part by the Science and Technology Innovation Program of Hunan Province (Grant No.2024RC1007), in part by the Young Scientists Fund of the National Natural Science Foundation of China (Grant 62303491), and in part by the Project of State Key Laboratory of Precision Manufacturing for Extreme Service Performance (Grant No.ZZYJKT2023-14).

\bibliography{aaai25}

\begin{thebibliography}{52}
\providecommand{\natexlab}[1]{#1}

\bibitem[{Andrychowicz et~al.(2016)Andrychowicz, Denil, Gomez, Hoffman, Pfau,
  Schaul, Shillingford, and De~Freitas}]{andrychowicz2016learning}
Andrychowicz, M.; Denil, M.; Gomez, S.; Hoffman, M.~W.; Pfau, D.; Schaul, T.;
  Shillingford, B.; and De~Freitas, N. 2016.
\newblock Learning to learn by gradient descent by gradient descent.
\newblock \emph{Advances in neural information processing systems}, 29.

\bibitem[{B{\l}asiok et~al.(2023)B{\l}asiok, Gopalan, Hu, and
  Nakkiran}]{blasiok2023unifying}
B{\l}asiok, J.; Gopalan, P.; Hu, L.; and Nakkiran, P. 2023.
\newblock A unifying theory of distance from calibration.
\newblock In \emph{Proceedings of the 55th Annual ACM Symposium on Theory of
  Computing}, 1727--1740.

\bibitem[{Blasiok and Nakkiran(2023)}]{blasioksmooth}
Blasiok, J.; and Nakkiran, P. 2023.
\newblock Smooth ECE: Principled Reliability Diagrams via Kernel Smoothing.
\newblock In \emph{The Twelfth International Conference on Learning
  Representations}.

\bibitem[{Br{\"o}cker(2012)}]{brocker2012estimating}
Br{\"o}cker, J. 2012.
\newblock Estimating reliability and resolution of probability forecasts
  through decomposition of the empirical score.
\newblock \emph{Climate dynamics}, 39: 655--667.

\bibitem[{Chidambaram et~al.(2024)Chidambaram, Lee, McSwiggen, and
  Rezchikov}]{chidambaramflawed}
Chidambaram, M.; Lee, H.; McSwiggen, C.; and Rezchikov, S. 2024.
\newblock How Flawed Is ECE? An Analysis via Logit Smoothing.
\newblock In \emph{Forty-first International Conference on Machine Learning}.

\bibitem[{Deng et~al.(2009)Deng, Dong, Socher, Li, Li, and
  Fei-Fei}]{deng2009imagenet}
Deng, J.; Dong, W.; Socher, R.; Li, L.-J.; Li, K.; and Fei-Fei, L. 2009.
\newblock Imagenet: A large-scale hierarchical image database.
\newblock In \emph{2009 IEEE conference on computer vision and pattern
  recognition}, 248--255. Ieee.

\bibitem[{Dong et~al.(2024)Dong, Jiang, Pan, Chen, Guan, Zhang, Gui, and
  Gui}]{dong2024survey}
Dong, J.; Jiang, Z.; Pan, D.; Chen, Z.; Guan, Q.; Zhang, H.; Gui, G.; and Gui,
  W. 2024.
\newblock A survey on confidence calibration of deep learning under class
  imbalance data.
\newblock \emph{Authorea Preprints}.

\bibitem[{Fernando and Tsokos(2021)}]{fernando2021dynamically}
Fernando, K. R.~M.; and Tsokos, C.~P. 2021.
\newblock Dynamically weighted balanced loss: class imbalanced learning and
  confidence calibration of deep neural networks.
\newblock \emph{IEEE Transactions on Neural Networks and Learning Systems},
  33(7): 2940--2951.

\bibitem[{Ferro and Fricker(2012)}]{ferro2012bias}
Ferro, C.~A.; and Fricker, T.~E. 2012.
\newblock A bias-corrected decomposition of the Brier score.
\newblock \emph{Quarterly Journal of the Royal Meteorological Society},
  138(668): 1954--1960.

\bibitem[{Geng et~al.(2024)Geng, Cai, Wang, Koeppl, Nakov, and
  Gurevych}]{geng2024survey}
Geng, J.; Cai, F.; Wang, Y.; Koeppl, H.; Nakov, P.; and Gurevych, I. 2024.
\newblock A Survey of Confidence Estimation and Calibration in Large Language
  Models.
\newblock In \emph{Proceedings of the 2024 Conference of the North American
  Chapter of the Association for Computational Linguistics: Human Language
  Technologies (Volume 1: Long Papers)}, 6577--6595.

\bibitem[{Grimmett and Stirzaker(2020)}]{grimmett2020probability}
Grimmett, G.; and Stirzaker, D. 2020.
\newblock \emph{Probability and random processes}.
\newblock Oxford university press.

\bibitem[{Guo et~al.(2017)Guo, Pleiss, Sun, and
  Weinberger}]{guo2017calibration}
Guo, C.; Pleiss, G.; Sun, Y.; and Weinberger, K.~Q. 2017.
\newblock On calibration of modern neural networks.
\newblock In \emph{International conference on machine learning}, 1321--1330.
  PMLR.

\bibitem[{Gupta et~al.(2020)Gupta, Rahimi, Ajanthan, Mensink, Sminchisescu, and
  Hartley}]{gupta2020calibration}
Gupta, K.; Rahimi, A.; Ajanthan, T.; Mensink, T.; Sminchisescu, C.; and
  Hartley, R. 2020.
\newblock Calibration of neural networks using splines.
\newblock \emph{arXiv preprint arXiv:2006.12800}.

\bibitem[{Han et~al.(2024)Han, Liu, Shang, Zheng, Zhong, Cao, Sun, and
  Xie}]{han2024balque}
Han, Y.; Liu, D.; Shang, J.; Zheng, L.; Zhong, J.; Cao, W.; Sun, H.; and Xie,
  W. 2024.
\newblock BALQUE: Batch active learning by querying unstable examples with
  calibrated confidence.
\newblock \emph{Pattern Recognition}, 110385.

\bibitem[{He et~al.(2016)He, Zhang, Ren, and Sun}]{he2016deep}
He, K.; Zhang, X.; Ren, S.; and Sun, J. 2016.
\newblock Deep residual learning for image recognition.
\newblock In \emph{Proceedings of the IEEE conference on computer vision and
  pattern recognition}, 770--778.

\bibitem[{Huang et~al.(2017)Huang, Liu, Van Der~Maaten, and
  Weinberger}]{huang2017densely}
Huang, G.; Liu, Z.; Van Der~Maaten, L.; and Weinberger, K.~Q. 2017.
\newblock Densely connected convolutional networks.
\newblock In \emph{Proceedings of the IEEE conference on computer vision and
  pattern recognition}, 4700--4708.

\bibitem[{Huang et~al.(2016)Huang, Sun, Liu, Sedra, and
  Weinberger}]{huang2016deep}
Huang, G.; Sun, Y.; Liu, Z.; Sedra, D.; and Weinberger, K.~Q. 2016.
\newblock Deep networks with stochastic depth.
\newblock In \emph{Computer Vision--ECCV 2016: 14th European Conference,
  Amsterdam, The Netherlands, October 11--14, 2016, Proceedings, Part IV 14},
  646--661. Springer.

\bibitem[{Huang et~al.(2020)Huang, Zhao, Zhu, Chen, and
  Broucke}]{huang2020experimental}
Huang, L.; Zhao, J.; Zhu, B.; Chen, H.; and Broucke, S.~V. 2020.
\newblock An experimental investigation of calibration techniques for
  imbalanced data.
\newblock \emph{Ieee Access}, 8: 127343--127352.

\bibitem[{Jiang et~al.(2012)Jiang, Osl, Kim, and
  Ohno-Machado}]{jiang2012calibrating}
Jiang, X.; Osl, M.; Kim, J.; and Ohno-Machado, L. 2012.
\newblock Calibrating predictive model estimates to support personalized
  medicine.
\newblock \emph{Journal of the American Medical Informatics Association},
  19(2): 263--274.

\bibitem[{Jiang et~al.(2023)Jiang, Dong, Pan, Wang, and Gui}]{zidonghua}
Jiang, Z.; Dong, J.; Pan, D.; Wang, T.; and Gui, W. 2023.
\newblock A novel intelligent monitoring method for the closing time of the
  taphole of blast furnace based on two-stage classification.
\newblock \emph{Engineering Applications of Artificial Intelligence}, 120:
  105849.

\bibitem[{Joy et~al.(2023)Joy, Pinto, Lim, Torr, and Dokania}]{sampleDependent}
Joy, T.; Pinto, F.; Lim, S.-N.; Torr, P.~H.; and Dokania, P.~K. 2023.
\newblock Sample-dependent adaptive temperature scaling for improved
  calibration.
\newblock In \emph{Proceedings of the AAAI Conference on Artificial
  Intelligence}, volume~37, 14919--14926.

\bibitem[{Krizhevsky, Hinton et~al.(2009)}]{krizhevsky2009learning}
Krizhevsky, A.; Hinton, G.; et~al. 2009.
\newblock Learning multiple layers of features from tiny images.

\bibitem[{Kull et~al.(2019{\natexlab{a}})Kull, Perello~Nieto, K{\"a}ngsepp,
  Silva~Filho, Song, and Flach}]{dirichlet}
Kull, M.; Perello~Nieto, M.; K{\"a}ngsepp, M.; Silva~Filho, T.; Song, H.; and
  Flach, P. 2019{\natexlab{a}}.
\newblock Beyond temperature scaling: Obtaining well-calibrated multi-class
  probabilities with dirichlet calibration.
\newblock \emph{Advances in neural information processing systems}, 32.

\bibitem[{Kull et~al.(2019{\natexlab{b}})Kull, Perello-Nieto, K{\"a}ngsepp,
  Song, Flach et~al.}]{kull2019beyond}
Kull, M.; Perello-Nieto, M.; K{\"a}ngsepp, M.; Song, H.; Flach, P.; et~al.
  2019{\natexlab{b}}.
\newblock Beyond temperature scaling: Obtaining well-calibrated multiclass
  probabilities with Dirichlet calibration.
\newblock \emph{arXiv preprint arXiv:1910.12656}.

\bibitem[{Kull, Silva~Filho, and Flach(2017{\natexlab{a}})}]{kull2017beta}
Kull, M.; Silva~Filho, T.; and Flach, P. 2017{\natexlab{a}}.
\newblock Beta calibration: a well-founded and easily implemented improvement
  on logistic calibration for binary classifiers.
\newblock In \emph{Artificial intelligence and statistics}, 623--631. PMLR.

\bibitem[{Kull, Silva~Filho, and Flach(2017{\natexlab{b}})}]{kull2017beyond}
Kull, M.; Silva~Filho, T.~M.; and Flach, P. 2017{\natexlab{b}}.
\newblock Beyond sigmoids: How to obtain well-calibrated probabilities from
  binary classifiers with beta calibration.
\newblock \emph{Electronic Journal of Statistics}, 11: 5052--5080.

\bibitem[{Kumar, Liang, and Ma(2019)}]{kumar2019verified}
Kumar, A.; Liang, P.~S.; and Ma, T. 2019.
\newblock Verified uncertainty calibration.
\newblock \emph{Advances in Neural Information Processing Systems}, 32.

\bibitem[{Lavine(1991)}]{lavine1991sensitivity}
Lavine, M. 1991.
\newblock Sensitivity in Bayesian statistics: the prior and the likelihood.
\newblock \emph{Journal of the American Statistical Association}, 86(414):
  396--399.

\bibitem[{Li and Caragea(2023)}]{li2023distilling}
Li, Y.; and Caragea, C. 2023.
\newblock Distilling calibrated knowledge for stance detection.
\newblock In \emph{Findings of the Association for Computational Linguistics:
  ACL 2023}, 6316--6329.

\bibitem[{Liu et~al.(2023)Liu, Rony, Galdran, Dolz, and
  Ben~Ayed}]{classAdaptive}
Liu, B.; Rony, J.; Galdran, A.; Dolz, J.; and Ben~Ayed, I. 2023.
\newblock Class Adaptive Network Calibration.
\newblock In \emph{Proceedings of the IEEE/CVF Conference on Computer Vision
  and Pattern Recognition}, 16070--16079.

\bibitem[{Lu, Lin, and Hu(2024)}]{lu2024disentangling}
Lu, Z.; Lin, R.; and Hu, H. 2024.
\newblock Disentangling Modality and Posture Factors: Memory-Attention and
  Orthogonal Decomposition for Visible-Infrared Person Re-Identification.
\newblock \emph{IEEE Transactions on Neural Networks and Learning Systems}.

\bibitem[{Luo et~al.(2024)Luo, Wu, Yang, Chen, Li, Peng, and
  Zhou}]{luo2024knowledge}
Luo, X.; Wu, J.; Yang, J.; Chen, H.; Li, Z.; Peng, H.; and Zhou, C. 2024.
\newblock Knowledge Distillation Guided Interpretable Brain Subgraph Neural
  Networks for Brain Disorder Exploration.
\newblock \emph{IEEE Transactions on Neural Networks and Learning Systems}.

\bibitem[{Mokhtari and Ribeiro(2020)}]{mokhtari2020stochastic}
Mokhtari, A.; and Ribeiro, A. 2020.
\newblock Stochastic quasi-newton methods.
\newblock \emph{Proceedings of the IEEE}, 108(11): 1906--1922.

\bibitem[{M{\"u}ller, Kornblith, and Hinton(2019)}]{muller2019does}
M{\"u}ller, R.; Kornblith, S.; and Hinton, G.~E. 2019.
\newblock When does label smoothing help?
\newblock \emph{Advances in neural information processing systems}, 32.

\bibitem[{Munir et~al.(2024)Munir, Khan, Khan, Ali, and
  Shahbaz~Khan}]{munir2024cal}
Munir, M.~A.; Khan, S.~H.; Khan, M.~H.; Ali, M.; and Shahbaz~Khan, F. 2024.
\newblock Cal-DETR: Calibrated Detection Transformer.
\newblock \emph{Advances in Neural Information Processing Systems}, 36.

\bibitem[{Naeini, Cooper, and Hauskrecht(2015)}]{naeini2015obtaining}
Naeini, M.~P.; Cooper, G.; and Hauskrecht, M. 2015.
\newblock Obtaining well calibrated probabilities using bayesian binning.
\newblock In \emph{Proceedings of the AAAI conference on artificial
  intelligence}, volume~29.

\bibitem[{Nelder and Mead(1965)}]{nelder1965simplex}
Nelder, J.~A.; and Mead, R. 1965.
\newblock A simplex method for function minimization.
\newblock \emph{The computer journal}, 7(4): 308--313.

\bibitem[{Nixon et~al.(2019)Nixon, Dusenberry, Zhang, Jerfel, and
  Tran}]{nixon2019measuring}
Nixon, J.; Dusenberry, M.~W.; Zhang, L.; Jerfel, G.; and Tran, D. 2019.
\newblock Measuring Calibration in Deep Learning.
\newblock In \emph{CVPR workshops}, volume~2.

\bibitem[{Patel et~al.(2020)Patel, Beluch, Yang, Pfeiffer, and
  Zhang}]{patelmulti}
Patel, K.; Beluch, W.~H.; Yang, B.; Pfeiffer, M.; and Zhang, D. 2020.
\newblock Multi-Class Uncertainty Calibration via Mutual Information
  Maximization-based Binning.
\newblock In \emph{International Conference on Learning Representations}.

\bibitem[{Penso, Frenkel, and Goldberger(2024)}]{penso2024confidence}
Penso, C.; Frenkel, L.; and Goldberger, J. 2024.
\newblock Confidence Calibration of a Medical Imaging Classification System
  that is Robust to Label Noise.
\newblock \emph{IEEE Transactions on Medical Imaging}.

\bibitem[{Platt et~al.(1999)}]{platt1999probabilistic}
Platt, J.; et~al. 1999.
\newblock Probabilistic outputs for support vector machines and comparisons to
  regularized likelihood methods.
\newblock \emph{Advances in large margin classifiers}, 10(3): 61--74.

\bibitem[{Rahimi et~al.(2020)Rahimi, Shaban, Cheng, Hartley, and
  Boots}]{intraOderPreserving}
Rahimi, A.; Shaban, A.; Cheng, C.-A.; Hartley, R.; and Boots, B. 2020.
\newblock Intra order-preserving functions for calibration of multi-class
  neural networks.
\newblock \emph{Advances in Neural Information Processing Systems}, 33:
  13456--13467.

\bibitem[{Roelofs et~al.(2022)Roelofs, Cain, Shlens, and
  Mozer}]{roelofs2022mitigating}
Roelofs, R.; Cain, N.; Shlens, J.; and Mozer, M.~C. 2022.
\newblock Mitigating bias in calibration error estimation.
\newblock In \emph{International Conference on Artificial Intelligence and
  Statistics}, 4036--4054. PMLR.

\bibitem[{Silva~Filho et~al.(2023)Silva~Filho, Song, Perello-Nieto,
  Santos-Rodriguez, Kull, and Flach}]{silva2023classifier}
Silva~Filho, T.; Song, H.; Perello-Nieto, M.; Santos-Rodriguez, R.; Kull, M.;
  and Flach, P. 2023.
\newblock Classifier calibration: a survey on how to assess and improve
  predicted class probabilities.
\newblock \emph{Machine Learning}, 112(9): 3211--3260.

\bibitem[{Wang et~al.(2024)Wang, Yang, Huang, and Cheng}]{wang2024moderate}
Wang, M.; Yang, H.; Huang, J.; and Cheng, Q. 2024.
\newblock Moderate Message Passing Improves Calibration: A Universal Way to
  Mitigate Confidence Bias in Graph Neural Networks.
\newblock In \emph{Proceedings of the AAAI Conference on Artificial
  Intelligence}, volume~38, 21681--21689.

\bibitem[{Wang and McCallum(2006)}]{wang2006topics}
Wang, X.; and McCallum, A. 2006.
\newblock Topics over time: a non-markov continuous-time model of topical
  trends.
\newblock In \emph{Proceedings of the 12th ACM SIGKDD international conference
  on Knowledge discovery and data mining}, 424--433.

\bibitem[{Zadrozny and Elkan(2001)}]{zadrozny2001obtaining}
Zadrozny, B.; and Elkan, C. 2001.
\newblock Obtaining calibrated probability estimates from decision trees and
  naive bayesian classifiers.
\newblock In \emph{Icml}, volume~1, 609--616.

\bibitem[{Zadrozny and Elkan(2002)}]{zadrozny2002transforming}
Zadrozny, B.; and Elkan, C. 2002.
\newblock Transforming classifier scores into accurate multiclass probability
  estimates.
\newblock In \emph{Proceedings of the eighth ACM SIGKDD international
  conference on Knowledge discovery and data mining}, 694--699.

\bibitem[{Zagoruyko and Komodakis(2016)}]{zagoruyko2016wide}
Zagoruyko, S.; and Komodakis, N. 2016.
\newblock Wide residual networks.
\newblock \emph{arXiv preprint arXiv:1605.07146}.

\bibitem[{Zellner(1996)}]{zellner1996models}
Zellner, A. 1996.
\newblock Models, prior information, and Bayesian analysis.
\newblock \emph{Journal of Econometrics}, 75(1): 51--68.

\bibitem[{Zhang, Kailkhura, and Han(2020)}]{mixNmatch}
Zhang, J.; Kailkhura, B.; and Han, T. Y.-J. 2020.
\newblock Mix-n-match: Ensemble and compositional methods for uncertainty
  calibration in deep learning.
\newblock In \emph{International conference on machine learning}, 11117--11128.
  PMLR.

\bibitem[{Zhang et~al.(2023)Zhang, Xie, Li, Mei, and Liu}]{zhang2023survey}
Zhang, X.-Y.; Xie, G.-S.; Li, X.; Mei, T.; and Liu, C.-L. 2023.
\newblock A survey on learning to reject.
\newblock \emph{Proceedings of the IEEE}, 111(2): 185--215.

\end{thebibliography}

\clearpage
\appendix
\section*{Appendix}
\renewcommand{\arraystretch}{2.0}
\section{Theory Proof}
\subsection{The Proof of Equivalence Solution}
\label{equivalence Solution}
\begin{proof}
	Let:
	\begin{equation}
		g_{p}(g(\hat S; \theta)) = C_{{N_{\hat S}}}^{N_{\hat S}^{pos}}g{(\hat S;\theta )^{N_{\hat S}^{pos}}}{(1 - g(\hat S;\theta ))^{{N_{\hat S}} - N_{\hat S}^{pos}}}.
	\end{equation}
	Then maximizing $g_{p}$ is equivalent to:
	\begin{equation}
		\frac{{\partial {g_p}}}{{\partial g}} = 0.
		\label{equi1}
	\end{equation}
	It is not difficult to calculate that the solution of Eq. \ref{equi1} is $g(\hat S;\theta)=\frac{{N_{\hat S}^{pos}}}{{{N_{\hat S}}}}$. Therefore, the essence of optimizing Eq. \ref{MHE_bin} is to find a suitable $\theta$ to minimize the distance between vector $(g({{\hat S}_1};\theta ),g({{\hat S}_2};\theta ), \cdots ,g({{\hat S}_{{N_b}}};\theta ))$ and vector $(\frac{{N_{{{\hat S}_1}}^{pos}}}{{{N_{{{\hat S}_1}}}}},\frac{{N_{{{\hat S}_2}}^{pos}}}{{{N_{{{\hat S}_2}}}}}, \cdots ,\frac{{N_{{{\hat S}_{{N_b}}}}^{pos}}}{{{N_{{{\hat S}_{{N_b}}}}}}})$, where $N_b$ represents the number of bins. Obviously, this is also the purpose of Eq. \ref{equivalence}.
\end{proof}

\subsection{The Proof of Theorem \ref{calibration_curve_lipschitz}}
\begin{proof}
	Let $\theta_{D}$ represent the parameters learned from data $D$. According to Eq. \ref{equivalence}, we have:
	\begin{equation}
		{\theta _D} = \mathop {{\mathop{\rm argmin}\nolimits} }\limits_\theta  \mathop {{\rm{ }}E}\limits_{\hat S} [{e^{{{(g(\hat S;\theta ) - \frac{{N_{\hat S}^{pos}}}{{{N_{\hat S}}}})^2}}}}].
	\end{equation}
	Let ${g_p}(\hat S;\theta ) = {e^{{{(g(\hat S;\theta ) - \frac{{N_{\hat S}^{pos}}}{{{N_{\hat S}}}})^2}}}}$. Obviously, $\frac{{\partial {g_p}}}{{\partial g}}$ is Lipchitz continuous $w.r.t.$ $g(\hat S;\theta )$ in $[0,1]$. Since ${g_p}(\hat S;\theta )$ is a convex function $w.r.t.$ $g(\hat S;\theta )$, $\frac{{\partial {g_p}}}{{\partial g}}$ is a monotone function $w.r.t.$ $g(\hat S;\theta )$. Therefore, there exists a function from $\frac{{\partial {g_p}}}{{\partial g}}$ to $g(\hat S;\theta )$, and $g(\hat S;\theta )$ is Lipchitz continuous $w.r.t.$ $\frac{{\partial {g_p}}}{{\partial g}}$. Therefore, $\forall$ a coupling $(\hat S_{1},\hat S_{2})$:
	\begin{equation}
		\begin{array}{l}
			\mathop E\limits_{({{\hat S}_1},{{\hat S}_2})} |g({{\hat S}_1};{\theta _{{D_1}}}) - g({{\hat S}_2};{\theta _{{D_2}}})|\\
			\le {L_1} \cdot \mathop E\limits_{({{\hat S}_1},{{\hat S}_2})} |\frac{{\partial {g_p}({{\hat S}_1};{\theta _{{D_1}}})}}{{\partial g({{\hat S}_1};{\theta _{{D_1}}})}} - \frac{{\partial {g_p}({{\hat S}_2};{\theta _{{D_2}}})}}{{\partial g({{\hat S}_2};{\theta _{{D_2}}})}}|.
		\end{array}
		\label{curve_lip_1}
	\end{equation}
	Similarly,  $\frac{{\partial {g_p}}}{{\partial g}}$ is Lipchitz continuous $w.r.t.$ $\frac{{N_{\hat S}^{pos}}}{{{N_{\hat S}}}}$ in $[0,1]$, then:
	\begin{equation}
		\begin{array}{*{20}{l}}
			{\mathop {{\rm{ }}E}\limits_{({{\hat S}_1},{{\hat S}_2})} |\frac{{\partial {g_p}({{\hat S}_1};{\theta _{{D_1}}})}}{{\partial g({{\hat S}_1};{\theta _{{D_1}}})}} - \frac{{\partial {g_p}({{\hat S}_2};{\theta _{{D_2}}})}}{{\partial g({{\hat S}_2};{\theta _{{D_2}}})}}|}\\
			\begin{array}{l}
				\le {L_2} \cdot \mathop {{\rm{ }}E}\limits_{({{\hat S}_1},{{\hat S}_2})} |\frac{{N_{{{\hat S}_1}}^{pos}}}{{{N_{{{\hat S}_1}}}}} - \frac{{N_{{{\hat S}_2}}^{pos}}}{{{N_{{{\hat S}_2}}}}}|\\
				\le {L_2} \cdot \mathop E\limits_{({{\hat S}_1},{{\hat S}_2})} [|\frac{{N_{{{\hat S}_1}}^{pos}}}{{{N_{{{\hat S}_1}}}}} - \frac{{N_{{{\hat S}_2}}^{pos}}}{{{N_{{{\hat S}_2}}}}}| + |{{\hat S}_1} - {{\hat S}_2}|].
			\end{array}
		\end{array}
		\label{curve_lip_2}
	\end{equation}
	Combining Eq. \ref{curve_lip_1} and Eq. \ref{curve_lip_2}, proof is completed.
\end{proof}

\subsection{The Proof of Theorem \ref{lipschitz}}
\label{lipschitz proof}
\begin{proof}
	We have:
	\begin{equation}
		TC{E_{bpm}}(D) = \int_0^1 {|g(\hat S;{\theta _D}) - \hat S| \cdot {\xi _D}(\hat S)} d\hat S,
	\end{equation}
	where $\theta _D$ represents the parameters learned from data $D$, and $\xi _D$ represents the probability density learned from data $D$. Then:
	
	\begin{equation}
		\begin{array}{l}
			|TC{E_{bpm}}({D_1}) - TC{E_{bpm}}({D_2})|\\
			= |\int_0^1 {|g(\hat S;{\theta _{{D_1}}}) - \hat S| \cdot {\xi _{{D_1}}}(\hat S) - |g(\hat S;{\theta _{{D_2}}}) - \hat S| \cdot {\xi _{{D_2}}}(\hat S)d\hat S} |\\
			= |\int_0^1 {|g(\hat S;{\theta _{{D_1}}}) - \hat S|}  \cdot {\xi _{{D_1}}}(\hat S) - |g(\hat S;{\theta _{{D_1}}}) - \hat S| \cdot {\xi _{{D_2}}}(\hat S)\\
			+ |g(\hat S;{\theta _{{D_1}}}) - \hat S| \cdot {\xi _{{D_2}}}(\hat S) - |g(\hat S;{\theta _{{D_2}}}) - \hat S| \cdot {\xi _{{D_2}}}(\hat S)d\hat S|\\
			\le \int_0^1 {|g(\hat S;{\theta _{{D_1}}}) - \hat S|}  \cdot |{\xi _{{D_1}}}(\hat S) - {\xi _{{D_2}}}(\hat S)|d\hat S\\
			+ |\int_0^1 {|g(\hat S;{\theta _{{D_1}}}) - g(\hat S;{\theta _{{D_2}}})|}  \cdot {\xi _{{D_2}}}(\hat S)d\hat S|\\
			\le \int_0^1 {|{\xi _{{D_1}}}(\hat S) - {\xi _{{D_2}}}(\hat S)|d\hat S} \\
			+ |\int_0^1 {|g(\hat S;{\theta _{{D_1}}}) - g(\hat S;{\theta _{{D_2}}})|}  \cdot {\xi _{{D_2}}}(\hat S)d\hat S|,
		\end{array}
		\label{theorem1_1}
	\end{equation}
	where the first inequality is because of the triangle inequality, and the second inequality is because of $|g(\hat S;{\theta _{{D_1}}}) - \hat S| \le 1$ holds for all $\hat S$. Because $g(\hat S; \theta_{D})$ and $\xi_{D}(\hat S)$ are Lipchitz continuous in data distributions for all $\hat S \in [0,1]$, then:
	\begin{equation}
		\begin{array}{l}
			|g(\hat S;{\theta _{{D_1}}}) - g(\hat S;{\theta _{{D_2}}})| \le {L_1} \cdot W({D_1},{D_2}),\\
			|{\xi _{{D_1}}}(\hat S) - {\xi _{{D_2}}}(\hat S)| \le {L_2} \cdot W({D_1},{D_2}).
		\end{array}
		\label{theorem1_2}
	\end{equation}
	Therefore, combining Eq. \ref{theorem1_1} and Eq. \ref{theorem1_2}, it holds:
	\begin{equation}
		\begin{array}{l}
			|TC{E_{bpm}}({D_1}) - TC{E_{bpm}}({D_2})|\\
			\le {L_2} \cdot W({D_1},{D_2}) \cdot \int_0^1 {1d\hat S}  + {L_1} \cdot W({D_1},{D_2}) \cdot |\int_0^1 {{\xi _{{D_2}}}(\hat S)d\hat S}| \\
			= ({L_1} + {L_2}) \cdot W({D_1},{D_2}).
		\end{array}
	\end{equation}
\end{proof}

\subsection{The Proof of Theorem \ref{Consistenty}}
\label{Consistenty proof}
\begin{proof}
	First, let's prove robust completeness. Let $\cal{P}$ be the family of all perfectly calibrated distributions. In Theorem \ref{lipschitz}, we have proved that $TCE_{bpm}(D)$ is Lipchitz continuous $w.r.t$ $D$. Therefore:
	\begin{equation}
		\forall \mathcal{D} \in \mathcal{P}, |TC{E_{bpm}}(D) - TC{E_{bpm}}(\mathcal{D})| \le L \cdot W(D,\mathcal{D}).
	\end{equation}
	Next, just prove that $TCE_{bpm}(\mathcal{D})=0$, robust completeness is proved. We know:
	\begin{equation}
		TC{E_{bpm}}(\mathcal{D}) = \int_0^1 {|g(\hat S;{\theta_\mathcal{D}}) - \hat S| \cdot {\xi_\mathcal{D}}(\hat S)} d\hat S.
	\end{equation}
	Because $ \mathcal{D} \in \mathcal{P}$, $P(H=1|\hat S) = \hat S$. Therefore:
	\begin{equation}
		TC{E_{bpm}}(\mathcal{D}) = \int_0^1 {|g(\hat S;{\theta_\mathcal{D}}) - P(H=1|\hat S)| \cdot {\xi_\mathcal{D}}(\hat S)} d\hat S.
	\end{equation}
	Let ${\theta ^*} = \mathop {{\mathop{\rm argmin}\nolimits} }\limits_{\theta  \in \Theta } {E_{\hat S}}|g(\hat S;\theta ) - P(H = 1|\hat S)|$,where $\Theta$ represents the parameter space of $g(\hat S;\theta )$. Due to $\cal{G}$ includes the calibration curve $P(H=1|\hat S)$, therefore $|g(\hat S;\theta^{*})-P(H=1|\hat S)|=0$. Then:
	\begin{equation}
	\begin{array}{*{20}{l}}
		{TC{E_{bpm}}({\cal D}) = \int_0^1 {|g(\hat S;{\theta _{\cal D}}) - P(H = 1|\hat S)| \cdot {\xi _{\cal D}}(\hat S)} d\hat S}\\
		\begin{array}{l}
			\le \int_0^1 {|g(\hat S;{\theta _{\cal D}}) - g(\hat S;{\theta ^*})| \cdot {\xi _{\cal D}}(\hat S)} d\hat S\\
			+ \int_0^1 {|g(\hat S;{\theta ^*}) - P(H = 1|\hat S)| \cdot {\xi _{\cal D}}(\hat S)} d\hat S
		\end{array}\\
		{ = \int_0^1 {|g(\hat S;{\theta _{\cal D}}) - g(\hat S;{\theta ^*})| \cdot {\xi _{\cal D}}(\hat S)} d\hat S.}
	\end{array}
	\end{equation}
	Let ${\cal D} = {\{ (\hat S,{N_{\hat S}},N_{\hat S}^{pos})\} _{0 \le \hat S \le 1}}$, and the ${D^*} = {\{ (\hat S,{N_{\hat S}},{N_{\hat S}} \cdot (P(H = 1|\hat S) + \frac{\varepsilon }{2}))\} _{0 \le \hat S \le 1}}$, where $D^{*}$ is the observation data of $g(\hat S;\theta^{*})$ and $\varepsilon/2$ is sampling error caused by random noise. Due to  $g(\hat S;\theta)$ is Lipschitz continuity $w.r.t.$ data distributions, then:
	\begin{equation}
	\begin{array}{*{20}{l}}
		{\int_0^1 {|g(\hat S;{\theta _{{\cal D}}}) - g(\hat S;{\theta ^*})| \cdot {\xi _{{\cal D}}}(\hat S)} d\hat S}\\
		{ \le L \cdot {E_{\hat S}}|\frac{{N_{\hat S}^{pos}}}{{{N_{\hat S}}}} - P(H = 1|\hat S) - \frac{\varepsilon }{2}|}\\
		{ \le L \cdot \mathop {\max }\limits_{\hat S} {{\{ |\frac{{N_{\hat S}^{pos}}}{{{N_{\hat S}}}} - P(H = 1|\hat S)|\} }_{0 < \hat S < 1}} + \frac{\varepsilon }{2},}
	\end{array}
	\end{equation}
	where $L$ is the Lipschitz constant. According to Hoeffding’s inequality, we have:
	\begin{equation}
	\begin{array}{l}
		P(|\frac{{N_{\hat S}^{pos}}}{{{N_{\hat S}}}} - P(H = 1|\hat S)| \ge \frac{\varepsilon }{2})\\
		= P(|N_{\hat S}^{pos} - {N_{\hat S}} \cdot P(H = 1|\hat S)| \ge {N_{\hat S}} \cdot \frac{\varepsilon }{2})\\
		\le {e^{ - \frac{1}{2} \cdot {N_{\hat S}} \cdot {\varepsilon ^2}}}.
	\end{array}
	\end{equation}
	Therefore, for $\forall \delta \in (0,1)$, when $N_{\hat S} \ge \frac{2In(1/\delta)}{\varepsilon^{2}}$, satisfy the following inequality with $1-\delta$ probability:
	\begin{equation}
	|\frac{{N_{\hat S}^{pos}}}{{{N_{\hat S}}}} - P(H = 1|\hat S)| \le \frac{\varepsilon }{2}.
	\end{equation}
	Therefore, for $\forall \delta \in (0,1)$, when $N_{\hat S} \ge \frac{2In(1/\delta)}{\varepsilon^{2}}$ holds for all $\hat S$, satisfy the following result with $1-\delta$ probability:
	\begin{equation}
	\begin{array}{l}
		TC{E_{bpm}}(\cal{D})\\
		\le L \cdot \mathop {\max }\limits_{\hat S} {\{ |\frac{{N_{\hat S}^{pos}}}{{{N_{\hat S}}}} - P(H = 1|\hat S)|\} _{0 < \hat S < 1}} + L \cdot \frac{\varepsilon }{2}\\
		\le L \cdot \frac{\varepsilon }{2} + L \cdot \frac{\varepsilon }{2}\\
		\le L \cdot \varepsilon  \to 0.
	\end{array}
	\label{robust_com}
	\end{equation}
	Therefore, robust completeness is proven. 
	
	Second, let's prove robust soundness. We know:
	\begin{equation}
	\begin{array}{*{20}{l}}
		{TC{E_{bpm}}(D) = \int_0^1 {|g(\hat S;{\theta _D}) - \hat S|}  \cdot {\xi _D}(\hat S)d\hat S}\\
		{ = \int_0^1 {|(P(H = 1|\hat S) - \hat S) - (P(H = 1|\hat S) - g(\hat S;{\theta _D}))|}  \cdot {\xi _D}(\hat S)d\hat S}\\
		{ \ge \int_0^1 {\left| {\left| {P(H = 1|\hat S) - \hat S} \right| - \left| {P(H = 1|\hat S) - g(\hat S;{\theta _D})} \right|} \right| \cdot {\xi _D}(\hat S)d\hat S} }\\
		{ \ge \left| {\int_0^1 {(\left| {P(H = 1|\hat S) - \hat S} \right|}  - \left| {P(H = 1|\hat S) - g(\hat S;{\theta _D})} \right|) \cdot {\xi _D}(\hat S)d\hat S} \right|}\\
		{ \ge \left| {\int_0^1 {\left| {P(H = 1|\hat S) - \hat S} \right|}  \cdot {\xi _D}(\hat S)d\hat S - L \cdot \left| \varepsilon  \right|} \right|}\\
		{ \ge \int_0^1 {\left| {P(H = 1|\hat S) - \hat S} \right|}  \cdot {\xi _D}(\hat S)d\hat S - L \cdot \left| \varepsilon  \right|,}
	\end{array}
	\label{soundness1}
	\end{equation}
	where the first inequality is due to the triangle inequality, the second inequality is due to Jensen's inequality, the third inequality is due to Eq. $\ref{robust_com}$, $L$ is Lipschitz constant, and $\varepsilon \to 0$. We have:
	\begin{equation}
	\int_0^1 {|P(H = 1|\hat S) - \hat S|}  \cdot {\xi _D}(\hat S)d\hat S = \mathop E\limits_D [|P(H = 1|\hat S) - \hat S|].
	\label{soundness2}
	\end{equation}
	Let $D = {\{ ({{\hat S}_i},{H_i})\} _{1 \le i \le N}} = {\{ ({{\hat S}_j},{N_{{{\hat S}_j}}},N_{{{\hat S}_j}}^{pos})\} _{1 \le j \le {N^{'}}}}$, where $N^{'}$ is the number of sampling locations of $\hat S$. Let ${{\cal D}^{'}} = {\{ (P(H = 1|{{\hat S}_j}),{N_{{{\hat S}_j}}},N_{{{\hat S}_j}}^{pos})\} _{1 \le j \le {N^{'}}}}$, Obviously, $\mathcal{D}^{'} \in \mathcal{P}$. Since $\hat S$ in $D$ and $\mathcal{D^{'}}$ is the same, then:
	\begin{equation}
	\begin{array}{l}
		\mathop E\limits_D [|P(H = 1|\hat S) - \hat S|] = \mathop E\limits_{{S_1} = {S_2}} [|P(H = 1|{{\hat S}_1}) - {{\hat S}_2}|\\
		= \mathop E\limits_{{S_1} = {S_2}} [|P(H = 1|{{\hat S}_1}) - {{\hat S}_2}| + |\frac{{N_{{{\hat S}_1}}^{pos}}}{{{N_{{{\hat S}_1}}}}} - \frac{{N_{{{\hat S}_2}}^{pos}}}{{{N_{{{\hat S}_2}}}}}|],
	\end{array}
	\label{soundness3}
	\end{equation}
	where $\hat S_{1} \in D$ and $\hat S_{2} \in \mathcal{D^{'}}$. Since $\hat S_{1}=\hat S_{2}$ is just a case of the couplings $(\hat S_{1},\hat S_{2})$, then:
	\begin{equation}
	\begin{array}{l}
		\mathop E\limits_{{S_1} = {S_2}} [|P(H = 1|{{\hat S}_1}) - {{\hat S}_2}| + |\frac{{N_{{{\hat S}_1}}^{pos}}}{{{N_{{{\hat S}_1}}}}} - \frac{{N_{{{\hat S}_2}}^{pos}}}{{{N_{{{\hat S}_2}}}}}|]\\
		\ge W(D,{{\cal D}^{'}}) \ge \underline{dCE}(D).
	\end{array}
	\label{soundness4}
	\end{equation}
	Combining Eq. $\ref{soundness1}$, Eq. $\ref{soundness2}$, Eq. $\ref{soundness3}$, and Eq. $\ref{soundness4}$, therefore:
	\begin{equation}
	TC{E_{bpm}}(D) \ge \underline{dCE}(D) - |\varepsilon|.
	\end{equation}
	Therefore, robust soundness is proven.
\end{proof}
\subsection{The Proof of Corollary \ref{corollary}}
\label{Corollary proof}
\begin{proof}
Because $P(\hat S|H=0)$ and $P(\hat S|H=1)$ follow beta distribution, according to Eq. \ref{prior_infer}, we can know that the hypothesis space $\cal{G}$ constructed by $g(\hat S;\alpha,\beta,c)$ includes the calibration curve $P(H=1|\hat S)$. According to Theorem \ref{Consistenty}, $TCE_{bpm}$ is a consistent calibration measure.
\end{proof}

\subsection{The Proof of Theorem \ref{efficiency}}
\label{efficiency_proof}
\begin{proof}
Before proving this theorem, we need to prove Lemma \ref{lemma1}.
\begin{lemma}
	For the function family:
	$$
	\begin{array}{l}
		{{\cal G}} = \{ g(\hat S;\alpha ,\beta ,c) = \frac{1}{{1 + {{\hat S}^{ - \alpha }}{{(1 - \hat S)}^\beta } \cdot {e^c}}}\\
		|\alpha  > 0,\beta  > 0,c \in ( - \infty , + \infty )\} ,
	\end{array}
	$$
	only three different observations $\{(\hat S^{*}_{1},g^{*}_{1}),(\hat S^{*}_{2},g^{*}_{2}),(\hat S^{*}_{3},g^{*}_{3})\}$ from any target unknown function $g(\hat S;\alpha^{*},\beta^{*},c^{*}) \in \mathcal{G}$ are needed to find $(\alpha^{*},\beta^{*},c^{*})$ by Algorithm \ref{Estimating calibration curve alg}.
	\label{lemma1}
\end{lemma}
\begin{proof}
 The solution process of Algorithm \ref{Estimating calibration curve alg} is equivalent to solving the following system of equations:
 \begin{equation}
 	\left\{ {\begin{array}{*{20}{c}}
 			{\frac{1}{{1 + \hat S{{_1^*}^{ - \alpha }}{{(1 - \hat S_1^*)}^\beta } \cdot {e^c}}} = g_1^*,}\\
 			{\frac{1}{{1 + \hat S{{_2^*}^{ - \alpha }}{{(1 - \hat S_2^*)}^\beta } \cdot {e^c}}} = g_2^*,}\\
 			\vdots \\
 			{\frac{1}{{1 + \hat S{{_3^*}^{ - \alpha }}{{(1 - \hat S_3^*)}^\beta } \cdot {e^c}}} = g_3^*.}
 	\end{array}} \right.
 	\label{eff1}
 \end{equation}
Equivalently transform Eq. $\ref{eff1}$ to get:
\begin{equation}
\left\{ {\begin{array}{*{20}{c}}
		{ - \alpha  \cdot \log (\hat S_1^*) + \beta  \cdot \log (\hat S_1^*) + c = \log (\frac{1}{{g_1^*}} - 1),}\\
		{ - \alpha  \cdot \log (\hat S_2^*) + \beta  \cdot \log (\hat S_2^*) + c = \log (\frac{1}{{g_2^*}} - 1),}\\
		\vdots \\
		{ - \alpha  \cdot \log (\hat S_3^*) + \beta  \cdot \log (\hat S_3^*) + c = \log (\frac{1}{{g_3^*}} - 1).}
\end{array}} \right.
\label{eff2}
\end{equation}
 Eq. $\ref{eff2}$ is a system of linear equations about $\alpha$, $\beta$, $c$. Due to $g(\hat S;\alpha^{*},\beta^{*},c^{*}) \in \mathcal{G}$, therefore Eq. $\ref{eff2}$ has a unique solution when the number of equations is greater than or equal to 3. That is, the number of different observations is greater than or equal to 3. 
\end{proof}
Let $g(\hat S;{\alpha ^*},{\beta ^*},{c^*}) = \mathop {{\mathop{\rm argmin}\nolimits} }\limits_{g \in G} {E_{\hat S}}[{e^{{{(g(\hat S;\alpha ,\beta ,c) - P(H = 1|\hat S))}^2}}}]$. According to Eq. \ref{prior_infer}, because $P(\hat S|H=0)$ and $P(\hat S|H=1)$ follow beta distribution, $E_{\hat S}|g(\hat S;{\alpha ^*},{\beta ^*},{c^*})-P(H=1|\hat S)|=0$. Then:
\begin{equation}
\begin{array}{*{20}{l}}
	{{E_{\hat S}}|g(\hat S;{\alpha _D},{\beta _D},{c_D}) - P(H = 1|\hat S)|}\\
	{ \le {E_{\hat S}}|g(\hat S;{\alpha _D},{\beta _D},{c_D}) - g(\hat S;{\alpha ^*},{\beta ^*},{c^*})|}\\
	{ + {E_{\hat S}}|P(H = 1|\hat S) - g(\hat S;{\alpha ^*},{\beta ^*},{c^*})|}\\
	{ = {E_{\hat S}}|g(\hat S;{\alpha _D},{\beta _D},{c_D}) - g(\hat S;{\alpha ^*},{\beta ^*},{c^*})|.}
\end{array}
\end{equation}
Constructing dataset ${D^*} = {\{ (\hat S,{N_{\hat S}},{N_{\hat S}} \cdot P(H = 1|\hat S))\} _{0 \le \hat S \le 1}}$ from $g(\hat S;{\alpha ^*},{\beta ^*},{c^*})$. Due to Theorem \ref{calibration_curve_lipschitz}, $g(\hat S;\alpha_{D},\beta_{D},c_{D})$ is Lipschitz continuous $w.r.t.$ $D$. Then:
\begin{equation}
\begin{array}{l}
	{E_{\hat S}}|g(\hat S;{\alpha _D},{\beta _D},{c_D}) - g(\hat S;{\alpha ^*},{\beta ^*},{c^*})|\\
	= {E_{\hat S}}|g(\hat S;{\alpha _D},{\beta _D},{c_D}) - g(\hat S;{\alpha _{{D^*}}},{\beta _{{D^*}}},{c_{{D^*}}})|\\
	\le L \cdot {E_{\hat S}}|\frac{{N_{\hat S}^{pos}}}{{{N_{\hat S}}}} - {\rm{P(H = 1|}}\hat S{\rm{)}}|.
\end{array}
\end{equation}
According to Hoeffding’s inequality, then:
\begin{equation}
\begin{array}{l}
P(|\frac{{N_{\hat S}^{pos}}}{{{N_{\hat S}}}} - P(H = 1|\hat S)| \ge \varepsilon )\\
= P(|N_{\hat S}^{pos} - {N_{\hat S}} \cdot P(H = 1|\hat S)| \ge {N_{\hat S}} \cdot \varepsilon ) \le {e^{ - 2 \cdot {N_{\hat S}} \cdot {\varepsilon ^2}}}.
\end{array}
\end{equation}
Therefore, $\forall \delta \in (0,1)$, when $N_{\hat S} \ge \frac{\rm{In}(1/\delta)}{2\varepsilon^{2}}$, satisfy the following result with $1-\delta$ probability:
\begin{equation}
|\frac{{N_{\hat S}^{pos}}}{{{N_{\hat S}}}} - P(H = 1|\hat S)| \le \varepsilon.
\end{equation}
Using union bound: $\forall \delta \in (0,1)$, when $N=N^{'}\cdot N_{\hat S} \ge \frac{N^{'}\cdot In(1/\delta)}{2\varepsilon^{2}}$, satisfy the following result with $1-\delta$ probability:
\begin{equation}
\begin{array}{l}
	L \cdot {E_{\hat S}}|\frac{{N_{\hat S}^{pos}}}{{{N_{\hat S}}}} - {\rm{P(H = 1|}}\hat S{\rm{)}}| \\
	\le L \cdot \frac{1}{{{N^{'}}}} \sum\limits_{i = 1}^{{N^{'}}} {|\frac{{N_{\hat S}^{pos}}}{{{N_{\hat S}}}} - {\rm{P(H = 1|}}\hat S{\rm{)}}|} \\
	\le L \cdot \varepsilon.
\end{array}
\end{equation}
Where $N^{'}$ the number of sampling locations of $\hat S$. According to Lemma 1, $g(\hat S;\alpha_{D},\beta_{D},c_{D})$ can be uniquely determined when $N^{'}\ge 3$. Therefore, the minimum sample size required is $\frac{3\cdot \rm{In}(1/\delta)}{2\varepsilon^{2}}$. Proof is completed.
\end{proof}

\section{Results}
\subsection{Experimental Data}
\label{Experimental data}
\subsubsection{Selection of publicly logit datasets}
In order to show the effect of the calibration curve estimated by the proposed method, this paper conducted experiments and effect visualization on ten public logit datasets \cite{roelofs2022mitigating}. The ten publicly available logit datasets arise from training four different architectures (ResNet, ResNet-SD, Wide-ResNet, and DenseNet) \cite{he2016deep, huang2017densely,zagoruyko2016wide,huang2016deep} on three different image datasets (CIFAR-10/100 and ImageNet) \cite{deng2009imagenet,krizhevsky2009learning}.

\subsubsection{Selection of True Distribution}
\label{selection}
When the true calibration curve distribution and confidence score distribution are known and Algorithm \ref{alg2} is used, a realistic calibration dataset, which obeys these true distributions, can be simulated. By using the simulated calibration dataset to estimate the calibration curve and then comparing it with the true calibration curve, the effect of the estimation can be observed. Similarly, by using the simulated calibration dataset to calculate existing calibration metrics and then comparing them with the known true calibration error, it becomes possible to determine whether the calibration metrics are good or bad.

To make the selected calibration curve and confidence score distribution representative, we refer to the fitting results of $Roelofs\ et\ al.$ \cite{roelofs2022mitigating}. They use the binary generalized linear model (GLM) to fit the calibration curve on public logit datasets and the Akaike Information Criteria (AIC) to select the optimal parameter model. Furthermore, they fit the confidence score distributions on public logit datasets using beta distributions.

We randomly selected five parameter models from their fitted parameter models as our true calibration curves and true confidence score distributions. The selected true distributions are shown in Table \ref{Selection of true distribution}, and the calibration curve shape of each true distribution is shown in Fig. \ref{plot_true_calibration_line}. There are obvious differences between these true calibration curves.
\begin{table}[H]
	\centering
	\footnotesize
	\setlength\tabcolsep{1pt}
	\renewcommand{\arraystretch}{1.3}
	\begin{tabular}{ccc}
		\toprule
		Number&Calibration curve&Confidence score\\
		\toprule
		D1&${\mathop{\rm logit^{-1}}\nolimits} ({\rm{ - 0}}.{\rm{88}} + 0.49 \cdot {\mathop{\rm logit}\nolimits} (\hat S))$&$Beta(2.77,0.04)$\\
		\hline
		D2&${\rm{logfli}}{{\rm{p}}^{ - 1}}({\rm{ - 0}}.12 + 0.58 \cdot {\rm{logflip}}(\hat S))$& $Beta(2.17,0.03)$\\
		\hline
		D3&${\log ^{ - 1}}( - 0.03 + 1.27 \cdot {\rm{log}}(\hat S))$& $Beta(1.12,0.11)$\\
		\hline
		D4&${{\mathop{\rm logit}\nolimits} ^{ - 1}}( - 0.77 - 0.80 \cdot {\rm{logflip(}}\hat S{\rm{)}})$& $Beta(1.13,0.20)$\\
		\hline
		D5&${{\mathop{\rm logit}\nolimits} ^{ - 1}}({\rm{ - 0}}.97 + 0.34 \cdot {\mathop{\rm logit}\nolimits} (\hat S))$& $Beta(1.19,0.22)$\\
		\bottomrule 
	\end{tabular}
\caption{Selection of true distribution, where $\rm{logflip} = \mathop{log}(1 - x)$. }
\label{Selection of true distribution}
\end{table}
\begin{figure}[h]
	\centering
	\includegraphics[width=0.4\textwidth]{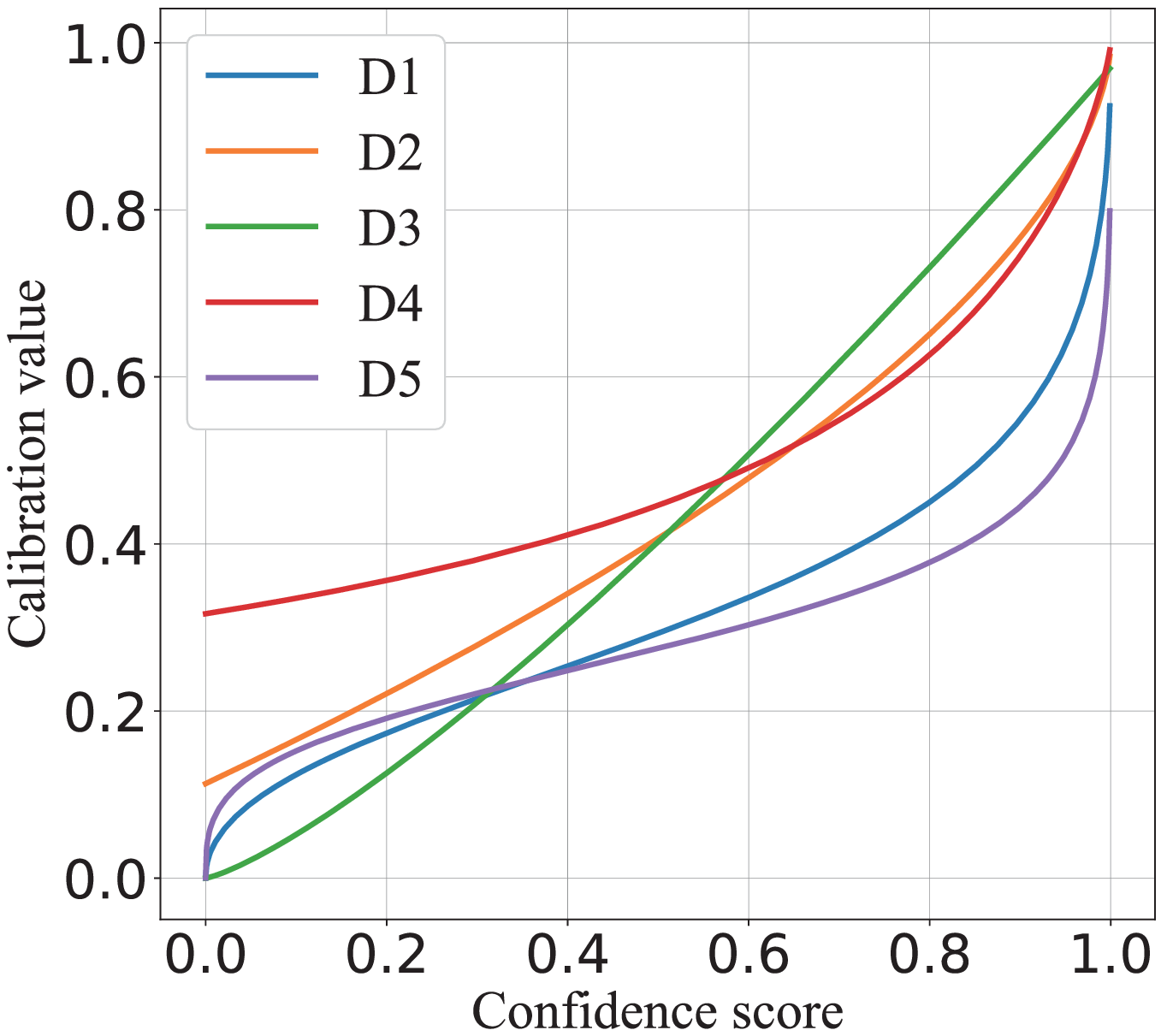}
	\caption{Visualization of the selected true calibration curve.}
	\label{plot_true_calibration_line}
\end{figure}

\subsection{Implementation Details}
\label{Hyperparameter settings}
In all experiments, $\mathcal{B}$ in Algorithm \ref{Estimating calibration curve alg} is set to a set of equal mass binnings. The number of bins is between $N/100, N/20]$, which ensures that each bin contains at least 20 samples and at most 100 samples. Nelder-Mead algorithm \cite{nelder1965simplex} is used to solve for the parameters of the calibration curve in maximum likelihood estimation. When calculating the histogram binning, equal mass binnings with bin numbers from 10 to 50 are considered. To measure the difference between the estimated calibration curve and the true calibration curve across $[0,1]$, the mean absolute difference is used, which is calculated as follows:
\begin{eqnarray}
	EAD= \frac{1}{1000} \sum_{i=0}^{1000} |g(\frac{i}{1000})-g^{*}(\frac{i}{1000})|,
	\label{EAD}
\end{eqnarray}
where $g(\cdot)$ represents the estimated calibration curve and $g^{*}(\cdot)$ represents the true calibration curve. $EAD$ differs from calibration metrics in that it does not consider the distribution of confidence scores but looks at each confidence point equally. Therefore, $EAD$ better measures the degree of fit across the overall curve and better captures the degree of fit in the low-density regions of confidence scores.

In addition to the naive ECE with equal mass binning, five state-of-the-art calibration metrics are selected to compare with the proposed calibration  metric $TCE_{bpm}$, as shown in Table \ref{Selection of comparison methods}. In the comparison of calibration metrics, to avoid the contingency caused by random errors, the final result is the average of 100 running results. For the calculation of binning-based calibration metrics, the number of bins is set to 15, which is a popular practice in the field \cite{guo2017calibration,roelofs2022mitigating}. This paper considers the comparison of calibration metrics under 10 cases with sample sizes ranging from 500 to 5000 (with a step size of 500).
\begin{table}[H]
	\centering
	\footnotesize
	\setlength\tabcolsep{0.5pt}
	\renewcommand{\arraystretch}{1.3}
	\begin{tabular}{ccc}
		\toprule
		Type&Name&Source\\
		\toprule
		\multirow{2}{*}{Binning-based}&$ECE_{debaised}$ \cite{kumar2019verified}&NeurIPS\\
		&$ECE_{sweep}$ \cite{roelofs2022mitigating}&AISTATS\\
		\hline
		\multirow{3}{*}{Binning-free}&$KS-error$ \cite{gupta2020calibration}&NeurIPS\\
		&$smECE$ \cite{blasioksmooth}&ICLR\\
		&$LS-ECE$ \cite{chidambaramflawed}&ICML\\
		\hline
		\bottomrule 
	\end{tabular}
\caption{Selection of comparison methods}
\label{Selection of comparison methods}
\end{table}

The hardware configuration used in these experiments includes Intel$^{\circledR}$ Core$^{TM}$ I7-10700 CPU, 3.70GHz, 125.5GB memory, NVIDIA Quadro RTX 5000 graphics card, 16GB of video memory. The software configuration for these experiments includes Ubuntu 20.04.3 LTS, Python 3.8.12, Torch 1.8.1+cu102, and Scipy 1.10.0.
\subsection{Estimated Results of Calibration Curves}
\label{Comparison in all confidence scores}
\subsubsection{Results in Real Datasets}
\label{Comparison in real datasets in all confidence scores}
Fig. \ref{fig_fit_real_data} shows the visualization of estimated calibration curves on nine public logit datasets. In each subplot of Fig. \ref{fig_fit_real_data}, the estimated calibration curve on real data aligns well with histogram binning results from various binning schemes and closely matches the mean result, which indicates that our method is relatively accurate and robust. From the length of the effective calibration result estimated by the histogram binning, it can be seen that there is different sharpness in the nine datasets. Nonetheless, our method achieves such relatively accurate and robust estimates in all these cases of sharpness. 
\begin{figure*}[h]
	\centering
	\includegraphics[width=1.\textwidth]{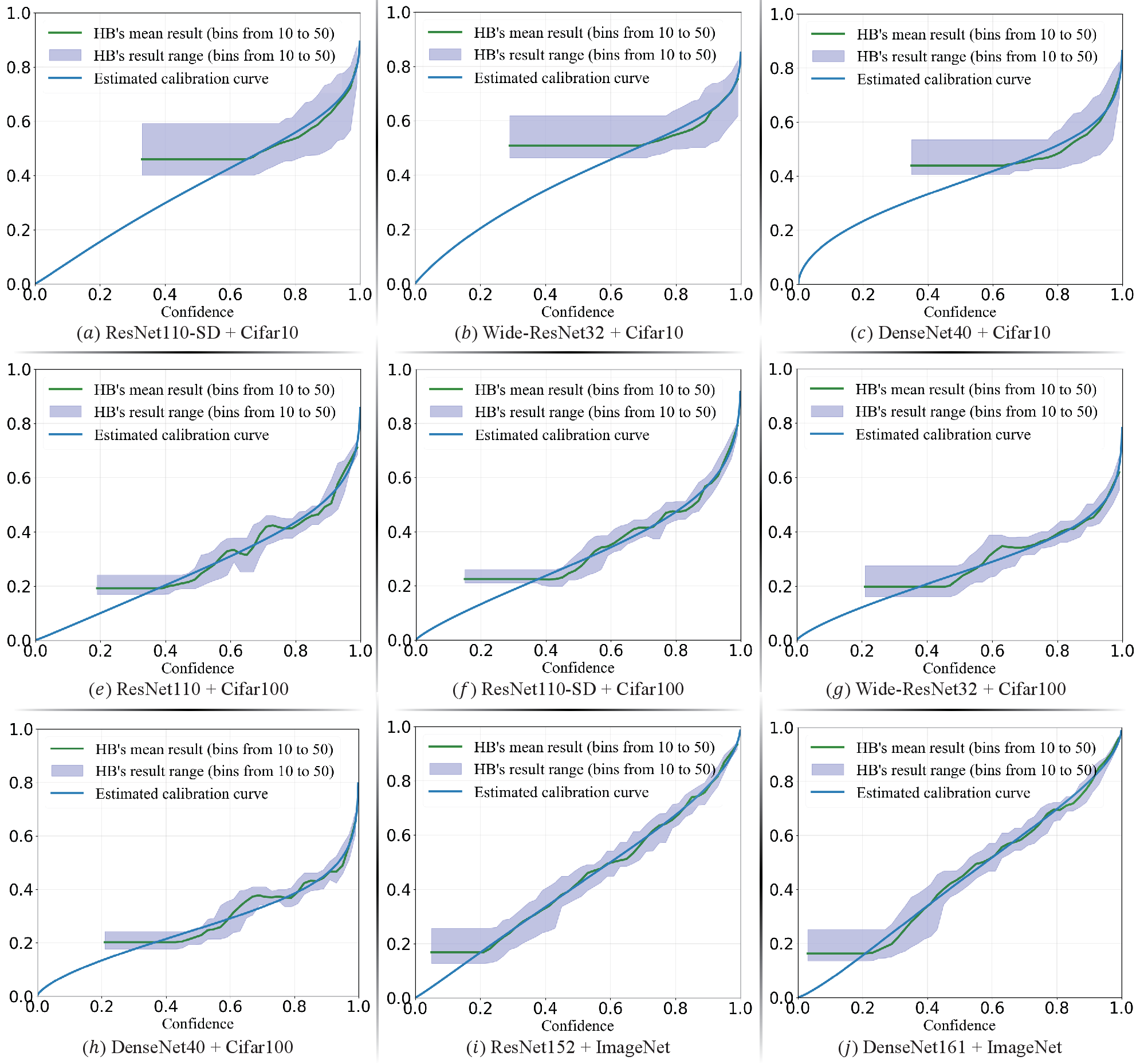}
	\caption{Visualization of the estimated calibration curve on the public logit dataset.}
	\label{fig_fit_real_data}
\end{figure*}
\subsubsection{Results in Simulating Datasets}
\label{Curve shape}
The estimated results of calibration curves on D2 to D5 in Table \ref{Selection of true distribution} are shown in Fig. \ref{fit_true_curve}. It can be seen that for each true calibration distribution, the calibration curve estimated by our method is closely aligned with the true calibration curve. This close alignment is also reflected in $EAD$ in Table \ref{Mean absolute error}, where $EAD$ of our method is much smaller than $EAD$ of the mean result of histogram binning. The $P$-values are much less than 0.01, which tells us that these decreases in $EAD$ are not due to chance. This validates the effectiveness of our method.
\begin{table}[h]
	\centering
	\footnotesize
	\setlength\tabcolsep{3pt}
	\renewcommand{\arraystretch}{0.8}
	\begin{tabular}{cccc}
		\toprule
		Number&Mean result of HB&Our method&$P$-value\\
		\toprule
		D1&0.0233 ($\pm$0.0047)&\textbf{0.0099} ($\pm$0.0062)&1.77$\times 10^{-15}$\\
		D2&0.1710 ($\pm$0.0200)&\textbf{0.0368} ($\pm$0.0196)&1.77$\times 10^{-15}$\\
		D3&0.0299 ($\pm$0.0037)&\textbf{0.0161} ($\pm$0.0063)&1.77$\times 10^{-15}$\\
		D4&0.0168 ($\pm$0.0042)&\textbf{0.0105} ($\pm$0.0048)&1.77$\times 10^{-15}$\\
		D5&0.0181 ($\pm$0.0030)&\textbf{0.0067} ($\pm$0.0029)&1.77$\times 10^{-15}$\\
		\bottomrule 
	\end{tabular}
\caption{Mean absolute error ($EAD$) of the calibration estimate, where results are expressed as mean ($\pm$standard) under 100 runs, bolding represents optimal results, and $P$-value represents the significance probability of Wilcoxon Signed-Rank Test.}
\label{Mean absolute error}
\end{table}
\begin{figure*}[h]
	\centering
	\includegraphics[width=1.\textwidth]{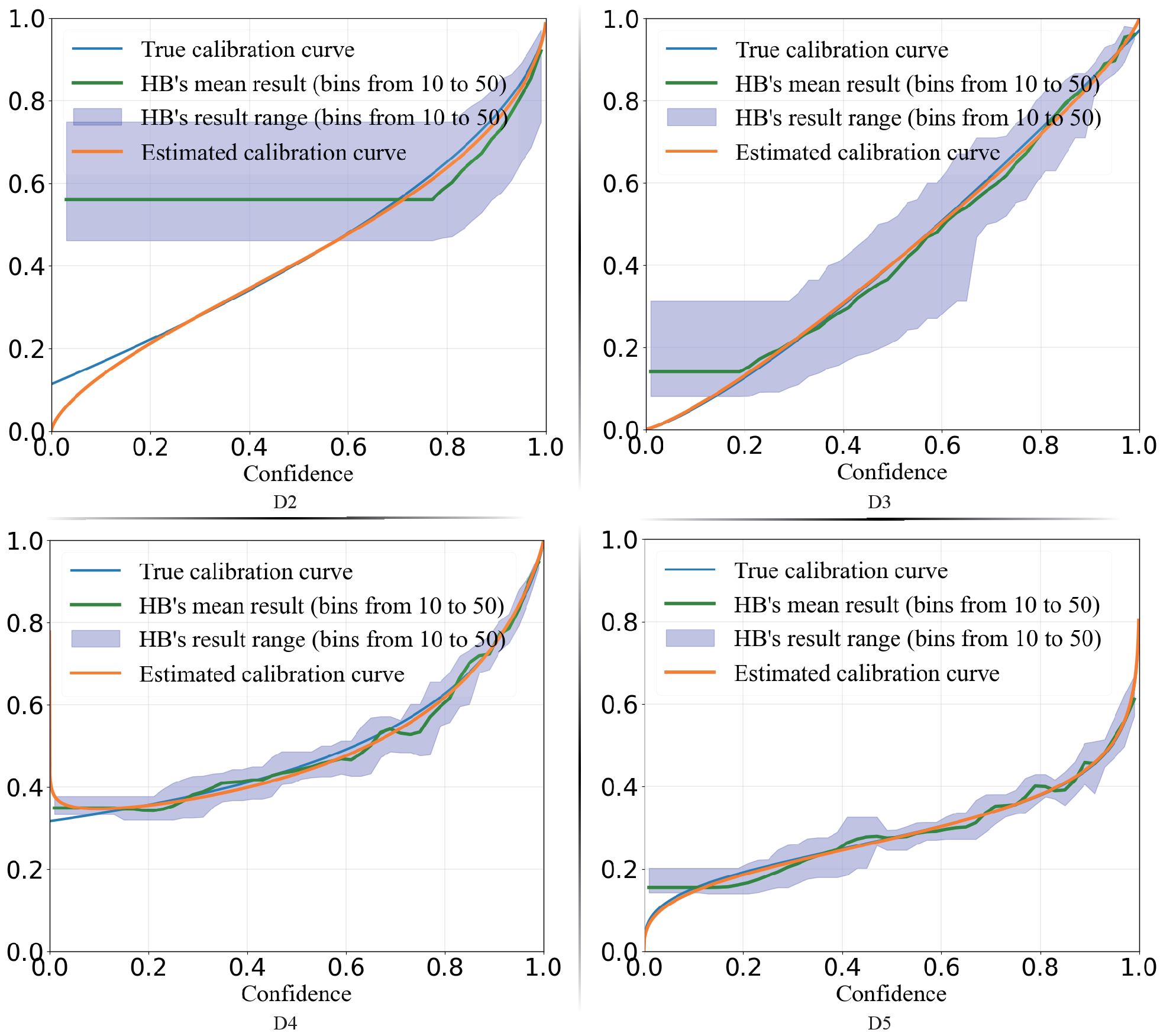}
	\caption{Visualization of estimated results of calibration curves.}
	\label{fit_true_curve}
\end{figure*}

\subsection{Estimated Results of Calibration Metrics}
\label{compare_metrics}
The calculation results of calibration metrics on the dataset simulated by the true distribution D2 to D5 are shown in Fig. \ref{metrics_in_all}. Fig. \ref{metrics_in_all} shows four different sizes of true calibration errors: around 1\%, around 7\%, around 10\%, and around 27\%, representing four degrees of miscalibration. It can be seen that all the selected calibration metrics have a relatively good estimate of the calibration error, with a difference of less than 2\% from the TCE. Sometimes, the results of multiple metrics even overlap. This is surprising because the principles behind these calibration metrics differ when they are proposed, but their results are sometimes relatively consistent. Overall, the results of $TCE_{bpm}$ are relatively stable and are closer to TCE in many cases. Especially in D3 of Fig. \ref{metrics_in_all}, $TCE_{bpm}$ almost completely outperforms other calibration metrics. This verifies that $TCE_{bpm}$ is effective and competitive.

\begin{figure*}[h]
	\centering
	\includegraphics[width=1.\textwidth]{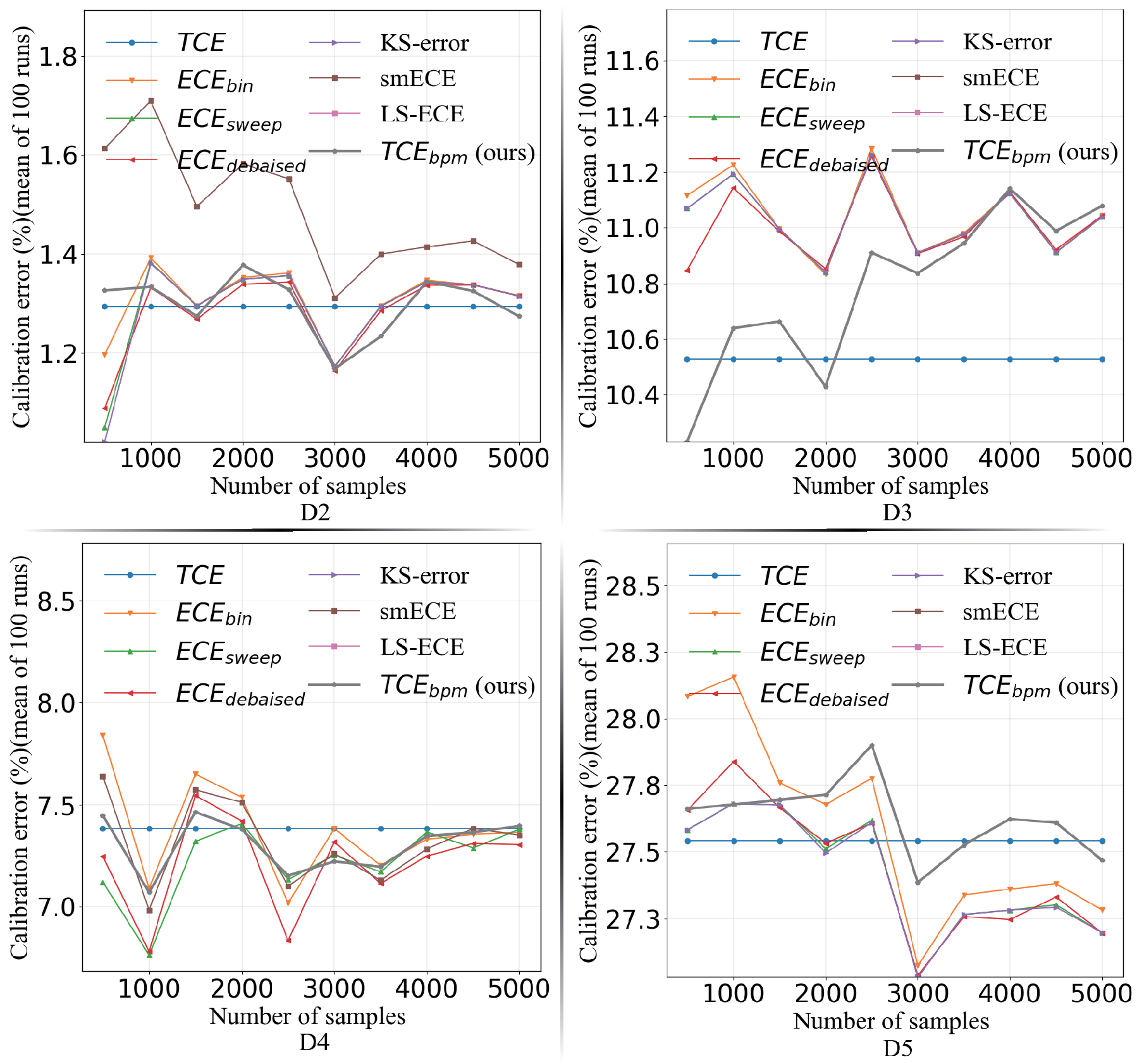}
	\caption{Calibration metrics comparison in all confidence scores.}
	\label{metrics_in_all}
\end{figure*}

\subsection{Comparison with Other Calibration Methods}
\label{effect of other calibration methods}
Table \ref{Calibration methods comparation} shows the comparison results of calibration metrics between our calibration method and other calibration methods on Wide-ResNet32's logits dataset of Cifar100 and DenseNet162 logits dataset of ImageNet. Comparing the calibrated and uncalibrated metrics shows that all considered calibration methods can significantly improve confidence. Obviously, our method outperforms all other methods in comparison on the two real datasets of different sizes.

\begin{table*}[h]
	\centering
	\footnotesize
	\setlength\tabcolsep{3pt}
	\renewcommand{\arraystretch}{1.3}
	\begin{tabular}{c|c|cccccc}
		\toprule
		Network and Dataset&Calibration methods&$ECE_{bin} \downarrow$&$ECE_{debaised} \downarrow$&$ECE_{sweep} \downarrow$&KS-error$\downarrow$&smECE$\downarrow$&$TCE_{bpm} \downarrow$\\
		\toprule
		\multirow{6}{*}{\makecell[c]{Wide-ResNet32\\ Cifar100}}&Uncalibration&0.18802&0.18794&0.18783&0.18784&0.17096&0.25405\\
		&Temperature scaling&0.01295&0.01013&0.01084&0.00861&0.01350&0.00997\\
		&Isotonic regression&0.01556&0.01670&0.01112&0.01130&0.01521&0.01694\\
		&Mix-n-Match&\underline{0.01161}&\underline{0.00787}&\underline{0.00944}&0.00741&\underline{0.01246}&\underline{0.00840}\\
		&Spline calibration&0.02170&0.02112&0.01838&\underline{0.00683}&0.01805&0.01136\\
		&TPM calibration (Ours)&\textbf{0.01118}&\textbf{0.00763}&\textbf{0.00899}&\textbf{0.00209}&\textbf{0.01216}&\textbf{0.00156}\\
		\hline
		\multirow{6}{*}{\makecell[c]{DenseNet162\\ ImageNet}}&Uncalibration&0.05722&0.05726&0.05719&0.05721&0.05620&0.06137\\
		&Temperature scaling&0.01966&0.01918&0.01913&0.00998&0.01916&0.01657\\
		&Isotonic regression&\underline{0.00745}&\underline{0.00401}&0.01069&0.00430&0.01202&0.00342\\
		&Mix-n-Match&0.01768&0.01725&0.01693&0.01061&0.01733&0.01472\\
		&Spline calibration&0.00880&0.00690&\underline{0.00692}&\underline{0.00190}&\underline{0.00882}&\underline{0.00230}\\
		&TPM calibration (Ours)&\textbf{0.00583}&\textbf{0.00307}&\textbf{0.00679}&\textbf{0.00188}&\textbf{0.00872}&\textbf{0.00042}\\
		\hline
		\bottomrule 
	\end{tabular}
\caption{Comparison with other calibration methods on real data. Bold represents the best result, and underline represents the second-best result.}
\label{Calibration methods comparation}
\end{table*}

\end{document}